\documentclass[journal]{IEEEtran}
\usepackage{hyperref}
\usepackage{amsmath}
\usepackage{graphics, subfigure, amsthm, times, amsfonts}
\usepackage{mathtools}
\usepackage{algorithm, algorithmic}
\usepackage{multirow}
\usepackage{wrapfig}
\usepackage{rotating}
\usepackage{bbm}

\newtheorem{assumption}{Assumption}
\newtheorem{theorem}{Theorem}
\newtheorem{lemma}{Lemma}
\newtheorem{definition}{Definition}
\newtheorem{remark}{Remark}
\newtheorem{proposition}{Proposition}

\title{A Rate Adaptation Algorithm for Tile-based 360-degree Video Streaming}
\vspace{-0.3in}
\author{Arnob Ghosh, Vaneet Aggarwal, and Feng Qian\thanks{A. Ghosh and V. Aggarwal are with the School of IE, Purdue University, West Lafayette IN 47907 (email: \{ghosh39, vaneet\}@purdue.edu). F. Qian is  with the Computer Science Department, Indiana University, Bloomington IN 47405 (email: fengqian@indiana.edu).}}
\date{}
\begin{document}
\maketitle
\begin{abstract}
In the 360-degree immersive video, a user only views a part of the entire raw video frame based on her viewing direction. However, today's 360-degree video players always fetch the entire panoramic view regardless of users' head movement, leading to significant bandwidth waste that can be potentially avoided. In this paper, we propose a novel adaptive streaming scheme for 360-degree videos. The basic idea is to fetch the invisible portion of a video at the lowest quality based on users' head movement prediction, and to adaptively decide the video playback quality for the visible portion based on bandwidth prediction. Doing both in a robust manner requires overcome a series of challenges, such as jointly considering the spatial and temporal domains, tolerating prediction errors, and achieving low complexity. To overcome these challenges, we first define quality of experience (QoE) metrics for adaptive 360-degree video streaming. We then formulate an optimization problem and solve it at a low complexity. The algorithm strategically leverages both future bandwidth and the distribution of users' head positions to determine the quality level of each tile (i.e., a sub-area of a raw frame). We further provide theoretical proof showing that our algorithm achieves optimality under practical assumptions. Numerical results show that our proposed algorithms significantly boost the user QoE by at least  20\% compared to baseline algorithms.

\end{abstract}%In 360 degree video, the user views a part of the entire raw video based on the head position. Fetching the entire raw video frame requires significant bandwidth overhead, which can be potentially reduced. In this paper, we consider the head position prediction and the available bandwidth  prediction to propose novel video streaming algorithms. We note that the inaccuracy in head position prediction (or the viewing area) leads to the viewer  seeing a dark spot which is detrimental for the Quality of Experience (QoE). Thus, there is a tradeoff between the quality at which the user views the viewing area and the amount of the raw video frame fetched to account for the prediction error which is the focus of the streaming algorithm. 

\begin{IEEEkeywords}
360-degree video, FoV estimation, Convex Optimization, Bandwidth Uncertainty, Stochastic Optimization.
\end{IEEEkeywords}
\section{Introduction}
\subsection{Motivation}
The {\em 360-degree} technology is shaping the video industry. 360-degree videos provide users a panoramic view creating a unique viewing experience. 360-degree videos, also known as immersive or spherical videos, are essential parts of the virtual reality (VR) which are changing the user's experience of video streaming. VR is projected to form a big market of \$120 billion by 2020 \cite{2020}.

 360-degree videos are recorded using the omnidirectional cameras. While watching the video, the user can change the viewing direction by changing the position of the head so that it can look for any location within the video. Typically, the user wearing a VR headset (e.g., the Google Cardboard) can adjust her orientation by changing the pitch, yaw, and roll of the device which corresponds to the X, Y and Z axes, respectively (Fig.~\ref{fig:360}). The filed-of-view (FoV) defines the extent of the user's observable portion. It is typically fixed for a VR headset (e.g., 90-degree vertically and 110-degree horizontally).  The video is divided into several chunks. %The FoV is a subset of the set of certain number of tiles which cover the viewing area for a certain co-ordinate.

\begin{figure}
\includegraphics[trim=3in 1.5in 1in 2.8in, clip, width=0.5\textwidth]{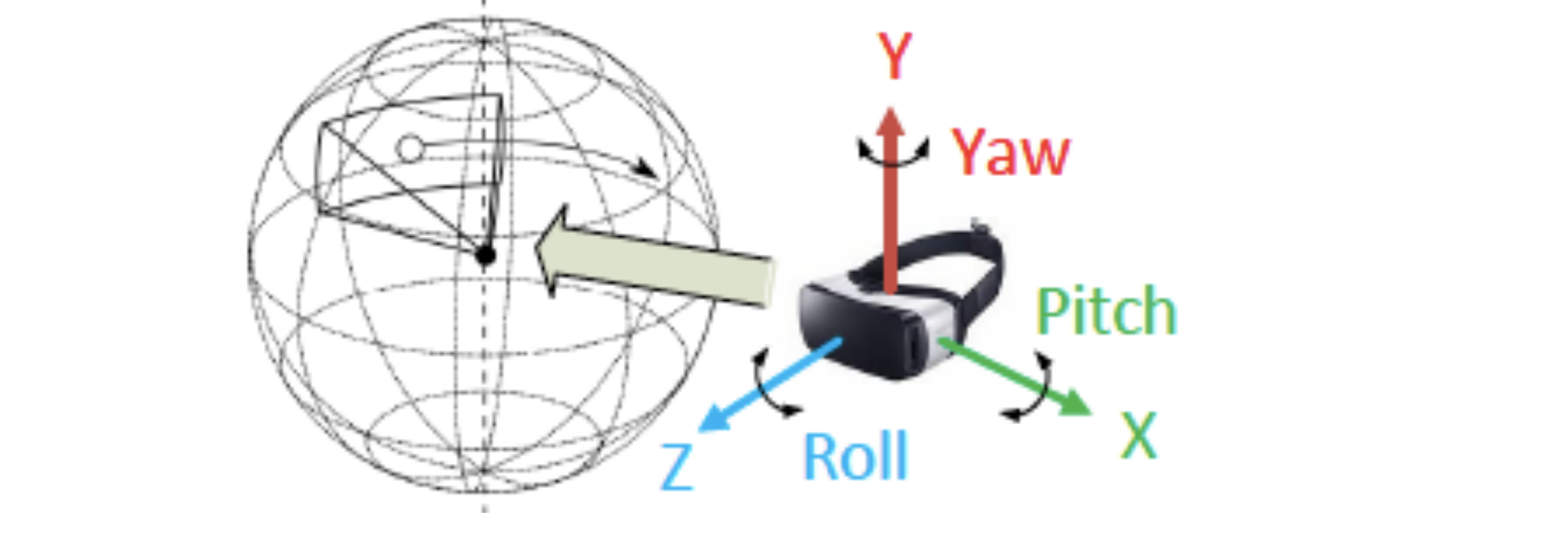}
\vspace{-0.2in}
\caption{Adjust viewing direction during 360 degree video playback}
\label{fig:360}
\vspace{-0.2in}
\end{figure}
 360-degree videos are very popular on major video platforms such as YouTube, Facebook.   However, the current popular technologies for streaming  try to fetch the all the portion of the chunk in the same quality including both the visible and invisible portions. Though this method is simple, it has some disadvantages. For example, the bandwidth utilization is high as the chunks in the 360-degree videos are of larger sizes compared to the traditional ones. Thus, if the network is congested or the bandwidth is low, it will lead to a poor quality video.  Hence, without smart algorithms, it can easily consume the wireless bandwidth\cite{bandwidthhungry}. Even the wireline capacity may not be enough for such 360-degree videos\cite{wireline}. Although significant progress has been made in developing VR technologies, the research community still lacks a bandwidth-efficient streaming algorithm for 360-degree videos for maximizing the quality of service (or quality of experience) of the users.

 \subsection{Our Contribution}
 We consider that the 360-degree video service provider wants to maximize the quality of experience (QoE) given the limited bandwidth. We assume that the user gets a utility depending on the rate at which a chunk is being downloaded. 
 %We consider that each chunk consists of several tiles. 
 The key idea of our approach  is that instead of downloading the entire panoramic view, a video player can download the portions of the chunk which is more likely to be viewed. 
This requires spatially pre-segmenting a 360-degree video chunk into multiple segments, which we call \emph{tiles}. A tile (as opposed to a chunk) is the smallest downloadable content unit in our scheme.
 The idea of downloading the visible portion of 360-degree videos in not new \cite{bao,feng}. However, none of these studies proposed a full-fledged streaming algorithm that adapts the streaming behavior at both the temporal domain (which quality to fetch for each chunk) and the spatial domain (which tiles to fetch for each chunk). Additionally, \cite{bao,feng} did not consider the FoV distribution and bandwidth estimation which can be leveraged to better guide the streaming algorithm. We therefore seek to contribute in this space. We assume that the distribution of the FoV can be obtained from the viewing history, the nature of the video, and even from the head movement of the user.\footnote{Recent papers \cite{bao,feng} shows that the head movement can be accurately predicted for a short duration.} Given the distribution of the FoV and the bandwidth estimation, our goal is to design a streaming algorithm that maximizes the QoE.  % We consider that each tile of a chunk can be encoded at different rates. Each chunk has to be downloaded before its deadline or play time, otherwise a stall will result. Since a user does not want any black spot in the viewing area, we assume that each tile must have to be downloaded at least at the minimum possible rate.  We also assume that the bandwidth can be estimated. Note that the bandwidth determines the download rates of the tiles as higher download rates may increase the stall time if there is not much available bandwidth.

In order to find an optimal algorithm for the QoE we have to first define the QoE metric. To the best of our knowledge, a concrete formulation of the QoE of a user in a 360-degree video is still missing. We, thus, begin with  providing the metric for QoE in the 360-degree video (Section~\ref{sec:qos}). The  stall duration ({\it i.e.}, the total time the playback's play-out buffer does not have any content to render) should be minimum\cite{yin}.    The user's viewing experience depends on the tiles' qualities within the FoV. If the quality of a tile within the chunk is too low, the quality will also be low.\footnote{Some encoding schemes such as Scalable Video Coding (SVC) can adapt to the lowest quality tile, by only displaying the lowest quality one in order to avoid quality variation.}  Hence, the user's QoE should be defined by the minimum rate among the tiles within the FoV. However, such a metric can be too pessimistic. For example, if many of the tiles except a few are of good quality in the FoV, the QoE may be high. Hence, we consider the expected viewing experience as the sum of the average playback rate of the visible tiles and the minimum of the rate among the tiles within the visible area. We maximize the above and meanwhile minimize the stall duration. %Hence, we consider the weighted sum of the sum of the playback bit rates of tiles within the FoV and the minimum rate among the tiles in the FoV as our {\em first} QoE  metric.

%

%Our objective is to maximize the above QoE metric and the negative of the stall time.  The weight for the stall time is kept very high to ensure that the stall time is always low.
Next, we formulate the problem as an optimization problem. However, it turns out to be non-convex because of the discrete variable space as the available download rates can only take discrete values. 
To address this challenge, we first relax the above constraint and formulate a relaxed problem which turns out to be convex. We then design a lightweight algorithm that computes the final scheduling decision where the selected bitrate belongs to a set of pre-defined discrete values (Section~\ref{sec:heuristic}). We also propose an online version of the above algorithm which can adapt the bandwidth or FoV variation dynamically.%We then provide a simple heuristic to obtain the feasible solution from the optimal solution of the relaxed problem which also does not increase the stall time compared to the optimal one. The key idea behind the heuristic is that we first obtain the maximum possible permissible rates not exceeded by the optimal solution of the relaxed problem, then, we assign additional bandwidth within each chunk to download higher quality tiles which have higher probabilities to constitute the FoV.

Although our first approach can maximize the expected QoE, it does not guarantee the bit-rate for the FoV with a high probability. If the FoV estimation has a lot of variances, the optimal solution which maximizes the expected QoE may render poor quality videos with a significant probability. \cite{feng,bao} also mentioned that the predicted head movement may not be accurate for a larger duration. Thus, as a second metric, we consider the QoE as the  bit-rate such that FoV will be at least that rate with probability $\alpha$ (Section~\ref{sec:robust}). The video service  provider wants to maximize the difference between the above rate and the weighted total stall time. We formulate the above problem as an optimization problem. %to maximize the above rate. %   the QoE is defined by the minimum rate  among the set of  tiles  which have very high probability to be the FoV. Hence, such an optimization problem will make sure that the FoV will be below that rate with a very low probability. %which constitute the event that the FoV will be within the set with a high probability.

 We propose an algorithm to solve the above problem.    We, also, provide a linear complexity heuristic algorithm which is optimal for a linear utility under the assumption that the FoV estimation is more accurate in the near future as compared to the distant future (Theorem~\ref{thm:optimal}). The above heuristic is similar to a greedy approach, where it will try to fetch the tiles  with a higher quality for a given chunk given that it does not exhaust the bandwidth and then will proceed to the subsequent chunks and so on. %We show that such a simple heuristic is optimal for a linear utility under the assumption that the FoV estimation is more accurate in the near future. We also propose another heuristic where in contrast to the above heuristic, it will fetch a higher quality tiles of a chunk if the similar quality tiles of other chunk can be downloaded within the estimated available bandwidth. We show that such a heuristic also provides a solution which is within the constant ratio of the optimal one, if the utilities do not increase much across the qualities.

We, numerically, evaluate the strength of our proposed algorithms compared to some readily implementable baseline algorithms (Section~\ref{sec:simulation}). We consider two types of baseline algorithms, the {\em first} algorithm downloads the entire chunk in the same quality. The second one is a greedy type of algorithm and similar to the ones proposed in \cite{bao,feng}. This algorithm only fetches the tiles which have the highest probability to be part of the FoV.  In both of these baseline algorithms, we assume that if the bandwidth permits they will first try to fetch the current chunks at  higher qualities rather the subsequent chunks. We show that the QoE is greatly enhanced using our proposed algorithms compared to the above baseline algorithms. As our algorithm achieves at least 20\% higher QoE compared to the baseline algorithms. We also show that our algorithm provides 20\% higher average bit-rate within the FoV. Our algorithm significantly outperforms the baseline algorithms when the uncertainty regarding the bandwidth and the FoV prediction increase.

To summarize, our main contributions are
\begin{itemize}
\item We formulate different QoE metrics for the 360-degree video streaming. In the first one, we maximize the expected QoE (Section~\ref{sec:qos}). The second one tries to minimize the probability that the user will see very low quality videos (Section~\ref{sec:robust}).
\item We formulate the optimization problem where the video service provider maximizes the above QoE metrics.
\item We provide low-complexity algorithms which maximize each of the proposed QoE metrics (Sections~\ref{sec:heuristic},\ref{sec:robust}).
\item We also provide a theoretical result that shows that such low-complexity algorithm is optimal under some assumptions which can arise frequently in practice (Theorem~\ref{thm:optimal}).
\item Empirically, we show that our proposed algorithms can achieve significantly higher QoE given the same bandwidth compared to the baseline algorithms (Section~\ref{sec:simulation}).
\end{itemize}
\subsection{Related Literature}Many recent studies have studied algorithms for online video streaming under limited bandwidth which maximizes the QoE \cite{yin,jiang}. However, these papers did not consider the {\em 360-degree videos}. The 360-degree videos propose unique challenges and thus require new metric for the QoE. For example,  in the 360-degree video each chunk consists of several tiles. Hence, a video streaming algorithm now needs to find the rate at which each tile of a chunk has to be downloaded. In contrast, the above papers only need to find the rate at which a chunk has to be downloaded. In the 360-degree video the FoV depends on the tiles of a chunk a user is viewing, the user may not view all the tiles, or in a viewing some tiles may be of different qualities which can impact the QoE. Hence, a new QoE metric is required for the 360-degree videos depending on the FoV.

 Heuristic based algorithms for 360-degree video streaming have been proposed \cite{hosseini, ganainy}. These papers proposed sending only those tiles which are part of the FoV. However, when the user changes his view, then he has to wait a certain amount of time for new tiles to be downloaded.  Recently, \cite{feng,bao} proposed FoV prediction based algorithms for the 360-degree video streaming. These papers estimate the FoV in the future  and fetch the tiles within that FoV. However, these papers did not provide any QoE metric for the 360-degree videos. Thus, there is no optimality guarantee for the heuristic. Second, the authors considered a sliding window protocol type protocol, where they predicted the FoV for a short time in the future and fetched the tiles in the predicted FoV. However, if there is an error in estimation, or the user views in some other areas, then there will be stall as the algorithm does not fetch anything outside the FoV. Third, the algorithms provided in \cite{feng,bao} did not consider the download rates for future chunks. Thus, if the future bandwidth is limited, the user may observe poor quality video in the future.

In contrast, we consider various metrics for the QoE and formulate as a stochastic optimization problem. We consider both the distribution of the FoV estimation and the bandwidth. We provide low-complexity algorithms which maximize the above QoE metrics. %We provide a heuristic which will be near the optimal solution.  Since we consider a optimization problem, thus, our algorithm balance between fetching future chunks and the current chunk. If the bandwidth is predicted to be limited in the future, the algorithm can fetch the future chunks instead of fetching higher quality tiles for the current chunk. We also provide an optimal solution and explicitly compute the above under some assumptions which frequently arise in practice. 
\section{System Model}
Suppose that the video consists of  $K$ chunks. We consider a non-real time video streaming, {\it i.e.}, a video has been recorded previously and needs to be streamed to the users as per their requests.\footnote{However, the algorithms we propose can be extended to the real time streaming with a minor modification.}
Each chunk $k=1,\ldots, K$ has $N$ tiles in total of the same duration. The tiles constitute a sub-area within a chunk.\footnote{Tile based segmentation is shown to be good in some applications. Nevertheless, our analysis and algorithm will also go through for other segmentation methods. However, a comparative analysis of different segmentations is a work for the future.} A tile has the same duration as the chunk. Each chunk is of duration $L$ seconds. Hence, the total duration of the video is $KL$. %\feng{Can we have a table that lists all notions?}

\textbf{Field of View}: A viewing area for chunk $k$ belongs to the set  $\mathcal{V}_k$ which is the set of sub-sets of $l$ tiles among all tiles. Thus, the field of view (FoV) $V_{f,k}\in \mathcal{V}_k$ consists of $l$ tiles of a chunk. The video player can choose to download tile $i\in \{1,\ldots, N\}$ of chunk $k$ at the rate $R_{i,k}$. Thus, $L\sum_{i}^{N}R_{i,k}$ is the size of the chunk\footnote{There can be a constant overhead for each chunk which we ignore. Our analysis/algorithm will remain the same even with the constant overhead.} $k$.
% if the tile $i$ of chunk $k$ is downloaded at the rate $R_{i,k}$

\textbf{Encoding Scheme}: %We consider the scalable video coding (SVC) for encoding the video\cite{}. Although much newer than many existing encoding schemes, SVC has existed for a while: it received the final approval to be standardized as an amendment of the H.264/MPEG-4 AVC standard in 2007\cite{}. The basic idea of SVC is to encode a chunk into ordered layers: one base layer (Layer 0) with the lowest playable quality, and multiple enhancement layers (Layer i >0) that further improve the chunk quality based on layer i ? 1. When downloading a chunk, a player must download all layers from 0 to i before fetching layer i + 1.
We consider that each chunk is encoded in different layers. Each layer corresponds to a bit-rate.
The bit-rate $R_{j}$ denotes the rate at the $j$-th layer. Thus, the tile $i$ within a chunk $k$ can be downloaded at rate  $R_{i,k}\in \{R_0,R_1,\ldots, R_{m}\}$ where $R_{m}>\ldots>R_1>R_0$. Hence, the set of possible rates is the same for each tile. 

Our approach can work both with the Adaptive Video Coding (AVC) \cite{avc}  or Scalable Video Coding (SVC) \cite{svc}.
%\feng{Change ``DASH'' to AVC (Adaptive Video Encoding) since DASH can work with either AVC or SVC. Mention the key difference between AVC and SVC is that in AVC, each video chunk is stored into independent encoding versions while in SVC the encoding of one version may depend on that of another version.}
 The main difference between AVC and SVC is that in AVC, each video chunk is stored into independent encoding versions while in SVC the encoding of one version may depend on that of another version. AVC is easy to implement and thus, it is most popular. However, SVC's unique encoding scheme has many advantages over the other state-of-the-art schemes.  For example,  if the chunk in AVC is not fully downloaded before its playback deadline, there will be a stall. However, this can be easily mitigated in the SVC, if a chunk can not be downloaded at layer $j+1$: it still can be played without any stall for layer up to $j$. 
We highlight the coding scheme that will be useful in each problem formulation or implementation of our proposed algorithms. %In the 360-degree videos, the FoV has to be estimated. However, if the actual view is different compared to the predicted FoV, the SVC can still play at the lowest possible rate if it is pre-fetched before. On top of this base layer, enhancement layers can be downloaded in the real time.  Also note that if there is a quality variation among different tiles within a FoV, using the SVC encoding scheme the video player can adjust to play at the minimum rate among the tiles which is not possible in the DASH. %The rate or quality with which the chunk $k$ is played depends on the quality of the tiles in a viewing area which has been downloaded. Hence, if the FoV is $V_{k}\in \mathcal{V}_k$, then the quality is $\min_{j\in V_{k}}R_{j,k}$. Even if the tiles within a FoV are downloaded at different rates, the video player can adjust to play at the minimum rate among the tiles.

 %as shown on the right side of Figure 1.As of today, SVC has not yet registered wide deployment in commercial video platforms likely due to itscomplexity and a lack of use cases. Also as expected, much less academic research has been conducted onSVC compared to DASH-style schemes.

% Hence,

\textbf{Download time and Bandwidth}: The video segments are downloaded into a {\em playback buffer} which contains downloaded but as yet unwatched video. Let $B(t)\in [0,B_{max}]$ denote the buffer occupancy at time $t$ , {\it i.e}., the play-time of the video left in the buffer at time $t$.  Chunk $k+1$ can only be downloaded once the chunk $k$ is downloaded. Once the chunk $k$ is downloaded, the video player waits a time $\zeta_k$ to start downloading the chunk $k+1$. We assume that $\zeta_k$ is small and it will not lead to the re-buffering events. Specifically, $\zeta_k$ is positive only when the buffer is full, otherwise, it is $0$. Let $t_k$ be the start time of downloading the chunk $k$. Hence,
\begin{align}\label{eq:tk}
t_{k}=t_{k-1}+\dfrac{L\sum_{i}R_{i,k}}{C_k}+\zeta_k,\quad t_0=0.
\end{align}
 where $C_k$ is the average bandwidth while downloading the chunk $k$, {\it i.e.}, in the interval $[t_k,t_{k+1}]$. {\em We assume that $C_k$ is known}.\footnote{Our model also allows the scenario where the bandwidth can be estimated on any time scale. If the time-scale overlaps with the download time $[t_k,t_{k+1}]$ one can easily get the total time by computing the time in each interval the bandwidth is estimated.} Note that even though the duration of a video can be a few minutes, the duration of chunk is only a few seconds. 
 Later, we show that for the online algorithm, we update the download rate for each chunk, hence, we need accurate prediction only for few seconds. The bandwidth can be predicted with high accuracy for such a short period \cite{proteus,vaneet}. Even if the mean of the bandwidth can not be accurately predicted, we will make a conservative approximation for the bandwidth and assume that $C_k$ be the worst possible bandwidth (or, the $C_k$ be the number such that the bandwidth during time $[t_k,t_{k+1}]$ will be higher than $C_k$ with a high probability).

\textbf{Play time of a chunk}: The buffer occupancy evolves as new chunk is downloaded. When a chunk is downloaded it is increased by $L$ and when the chunk is played, it decreases by $L$. We denote the play-time of the $k$th chunk as $\tilde{t}_k$, i.e., the $\tilde{t}_k$ is the time when the $k$-th chunk starts playing. Note that the $k$-th chunk can only start playing only when the k-th chunk is downloaded. Hence, for chunk $k>1$,
\begin{align}\label{b_k}
\tilde{t}_{k}=\max\{\tilde{t}_{k-1}+L,t_{k}+\dfrac{L\sum_{i}R_{i,k}}{C_k}\}, \quad \tilde{t}_1=t_{\mathrm{ini}}.
\end{align}
where $t_{\mathrm{ini}}$ is the initial start-up time or initial stall time. The initial start-up time is often considered to be constant.
%The buffer occupancy evolves as new chunk is downloaded. When a chunk is downloaded it is increased by $L$ and when the  chunk is played, it decreases. The evolution of the buffer occupancy is characterized in the following--
%\begin{align}\label{b_k}
%B_{k+1}=((B_k-\dfrac{d_k(\sum_{i=1}^{N_k}R_{i,k})}{C_k})^{+}+L-\alpha_k)^{+}.
%\end{align}

\textbf{Stall time}: Note that each chunk constitutes $L$ amount of time of the original video. Hence, if $\tilde{t}_{k}>(k-1)L+t_{\mathrm{ini}}$, then, there will be stall or re-buffering. Thus, the total stall time is $(\tilde{t}_K-(K-1)L-t_{\mathrm{ini}})^{+}$.  Note from (\ref{b_k}) that $\tilde{t}_{K}\geq (K-1)L+t_{\mathrm{ini}}$.  Thus,
\begin{align}\label{stall}
(\tilde{t}_K-(K-1)L-t_{\mathrm{ini}})^{+}=\tilde{t}_K-(K-1)L-t_{\mathrm{ini}}.
\end{align}
%$\dfrac{d_k(\sum_{i=1}^{N_k}R_{i,k})}{C_k}>B_k$, the buffer becomes empty while still downloading the chunk $k$. Hence, there will be stall or re-buffering.

\textbf{Maximum Buffer Occupancy}: As mentioned before we assume that as soon as chunk $k$ is downloaded, the video player starts downloading the chunk $k+1$. Hence, $\zeta_k$ is $0$ when the buffer is not full. However one exception is when the buffer is full, the player waits for the buffer to reduce to a level which allows downloading the next chunk.

We also assume that the maximum buffer occupancy is a multiple of $L$.\footnote{It is straight forward to extend to the setting when the maximum buffer occupancy is not a multiple of $L$.} The buffer can store at most $B$ chunks. Hence, we must have
\begin{align}\label{eq:max_buffer}
t_{B+k}\geq \tilde{t}_k\quad k=1,\ldots, K-B.%\alpha_k=((B_k-\dfrac{d_k(\sum_{i=1}^{N_k}R_{i,k})}{C_k})^{+}+L-B_{max})^{+}.
\end{align}
$\zeta_k$ is adjusted in (\ref{eq:tk}) such that the above constraint is satisfied. Currently, the maximum buffer occupancy $B$ can be very high, and, thus, the above constraint is almost always satisfied with $\zeta_k=0$ for all $k$.

\textbf{User's utility}
The user obtains an utility $U(\cdot)$ depending on the rate at which the video is being played. For example, if the chunk is played at rate $R$, the user's utility is $U(R)$ which denotes the user's satisfaction for getting the chunk at rate $R$. $U(\cdot)$ is a strictly increasing function as the user strictly prefers a higher rate. 

\section{Quality of Experience}\label{sec:qos}
The user's quality of experience (QoE) depends on the rate at which the tiles are downloaded in a viewing area. If a tile is  played at a poor quality, the overall user experience may decrease even though the rest of the tiles are of good qualities. Hence, for such a user, if the FoV is $V_{f,k}\in \mathcal{V}_k$, then, the utility for the user for chunk $k$ is
\begin{align}\label{eq:pess}
U(\min_{i\in V_{f,k}}R_{i,k}).
\end{align}
Recall that $V_{f,k}$ is the set of tiles of $l$ tiles. 

One option for the streaming service is to display all the tiles in the FoV at the same rate. This is where the SVC has an advantage since irrespective of the download rate of each tile, the tiles can be played at the lowest rate of the tiles in a viewing area ($\min_{i\in V_f}R_{i,k}$) by discarding the higher layers (enhancement layers).
Since $U(\cdot)$ is a strictly increasing function, the following holds.

However, the expression in (\ref{eq:pess}) may be too pessimistic approach to judge the quality. For example, if most of the tiles in the field of view are of good quality, the one with a slightly worse condition may not degrade the QoE much. Hence, we consider the QoE as the following for chunk $k$ if the FoV is $V_{f,k}\in \mathcal{V}_k$.
\begin{align}
U(\min_{i\in V_{f,k}}R_{i,k})+\gamma \sum_{i\in V_{f,k}}U(R_{i,k})
\end{align}
where $\gamma\geq 0$.  Note that the second term corresponds to the sum of the qualities of the tiles within the FoV. If $\gamma=0$, then the user's QoE is governed by the minimum rate of tiles in a viewing area. As $\gamma$ increases, the weight of the lowest rate tile decreases, and there will be more chances that the tiles within a FoV are of varying qualities.  Thus $\gamma$  is a tradeoff parameter that can be chosen for trading off between the minimum tile quality and the average tile quality.

\subsection{QoE maximization Problem}
The goal is to maximize the QoE and minimize the total stall time (cf. (\ref{stall})). Assuming that the FoV is known exactly beforehand the  optimization problem is
\begin{eqnarray}
\text{maximize } & \sum_{k=1}^{K}[U(\min_{i\in V_{f,k}}R_{i,k})+\gamma \sum_{i\in V_{f,k}}U(R_{j,k})]\nonumber\\& -\lambda(\tilde{t}_{K}-(K-1)L-t_\mathrm{ini})\nonumber\\
\text{subject to }&  (\ref{b_k}),(\ref{eq:tk}), (\ref{eq:max_buffer})\nonumber\\
& R_{i,k}\in \{R_0,\ldots, R_m\}\nonumber\\
\text{var}: & R_{i,k}\nonumber
\end{eqnarray}
where $V_{f,k}\in \mathcal{V}_k$ is the FoV for chunk $k$. The second term of the objective corresponds to the penalty associated with the total stall time.

\textbf{FoV Estimation}: Note that in the above formulation we assume that the FoV is known exactly before hand. However, in practice it may not be known beforehand. Recently, \cite{bao,feng} developed techniques to estimate the FoV for a shorter duration ahead. {\em We assume that the distribution of the FoV of each chunk is known beforehand.} This distribution can also be achieved from the other user's viewing history. Note that in the online version, we consider a sliding window type algorithm where we only need FoV estimation for few seconds in the future. \cite{bao,feng} showed that such a distribution of the FoV can be achieved for that short duration. \footnote{The consideration of any specific prediction method  is beyond the scope of this paper and left for the future.}

Since $U(\cdot)$ is a strictly increasing function, thus, 
\begin{lemma}
$U(\min_{i\in V_{f,k}}R_{j,k})=\min_{i\in V_{f,k}}U(R_{i,k})$.
\end{lemma}
Hence,
\begin{align}
\mathbbm{E}[U(\min_{i\in V_{f,k}}R_{i,k})]=\mathbbm{E}[\min_{i\in V_{f,k}}U(R_{i,k})].
\end{align}

 We maximize the expected QoE. Hence, using the above the optimization problem becomes
\begin{eqnarray}
\mathcal{P}: \text{max } & \sum_{k=1}^{K}\mathbbm{E}[\min_{i\in V_{f,k}}U(R_{i,k})+\gamma \sum_{i\in V_{f,k}}U(R_{i,k})]\nonumber\\& -\lambda(\tilde{t}_K-(K-1)L-t_{\mathrm{ini}})\label{eq:pe}\\
\text{subject to }&  (\ref{b_k}),(\ref{eq:tk}), (\ref{eq:max_buffer})\nonumber\\
\text{var}: & R_{i,k}\in \{R_0,\ldots, R_{m}\}\nonumber
\end{eqnarray}
The expectation is taken over the distribution of the FoV. %The uncertainty also arises because of the bandwidth uncertainty, however, we consider that the $C_k$ is estimated as the mean. %Note that in the online algorithm, we do not optimize over the  once we download later chunks, hence, $t_{0}$ is considered to be $0$ for later chunks.
%\feng{``however...'': I don't understand the logic here.}

\begin{remark}
In general, the users do not like  stalls, hence,  $\lambda$ is assumed to  be  large.  As mentioned earlier, $\gamma$ has to be chosen judiciously as it determines the weight of different aspects of the QoE.
\end{remark}
\begin{remark}
Note that the decision space is discrete, hence, the problem is non-convex. Thus, we can not apply the standard convex optimization techniques to find the optimal solution. 
\end{remark}
\begin{remark}
In general, the lowest rate $R_0$ is the base layer rate. One may want to choose $R_0=0$ since it may be fine not to download and display some tiles. However, this may result in dark spots in part of the view or complete view which may not be preferable. Thus, we consider that a tile has to be downloaded at least at the base layer rate $R_0>0$ which implies that the user will be able to watch videos at least at the base layer rate in any viewing area.
%If a tile is fetched at rate $0$, the tile can not be shown (or, show a dark spot). However, in the 360 degree video, an user may see want to see every possible direction.
\end{remark}
\begin{remark}
In the optimization problem, we maximize the QoE in an expected sense. However, one of the drawbacks of the above approach is that if the FoV has a large variance,  the observed quality may still be low with a significant probability. However, in Section~\ref{sec:robust} we consider a robust approach where we address the above issue. %want to make sure that the quality variation will be low with a very high probability.

\end{remark}
\subsection{An equivalent Problem}
The constraint in (\ref{b_k}) is not convex.  In the following we represent (\ref{b_k}) in an equivalent convex form.
\begin{theorem}
Represent the constraint in (\ref{b_k}) as the following
\begin{align}\label{eq:equi}
\tilde{t}_k\geq \tilde{t}_{k-1}+L, \quad\tilde{t}_k\geq t_{k}+\dfrac{\sum_{i}R_{i,k}}{C_k}, \quad
\tilde{t}_1\geq t_{\mathrm{ini}}.
\end{align}
Now, consider the following problem
\begin{eqnarray}
\mathcal{P}_{e}: \text{maximize } & (\ref{eq:pe})\nonumber\\
\text{subject to }&  (\ref{eq:equi}), (\ref{eq:tk}), (\ref{eq:max_buffer})\nonumber\\
\text{var}: & R_{i,k}\in \{R_0,\ldots, R_{m}\}\nonumber
\end{eqnarray}Then, $\mathcal{P}_{e}$ is equivalent to $\mathcal{P}$.
\end{theorem}
{\em Proof}:
Note that (\ref{eq:equi}) is not equivalent to (\ref{b_k}). In fact there are $\tilde{t}_k$ which satisfy (\ref{eq:equi}), however, they do not satisfy (\ref{b_k}). For example, if $\tilde{t}_k>\max\{\tilde{t}_{k-1}+L,t_{k}+L\dfrac{\sum_{i}R_{i,k}}{C_k}\}$, then, it will not satisfy (\ref{b_k}) even though it satisfies (\ref{eq:equi}). However, note that the objective function is strictly decreasing in $\tilde{t}_k$. Hence, in the optimal solution, we will have $\tilde{t}_k=\max\{\tilde{t}_{k-1}+L,t_{k}+\dfrac{L\sum_{i}R_{i,k}}{C_k}\}$. Thus, any optimal solution of $\mathcal{P}$ is also an optimal solution of $\mathcal{P}_e$ and vice versa which proves the result as in the statement of the Theorem.  \qedsymbol

However, the {\em above problem is NP-hard} in general because of the discrete strategy space of $R_{i,k}$. In the following, we consider a relaxed version which is computationally easy to solve.
\subsection{Relaxation}
Now, we provide a relaxed problem which is convex for concave utility functions. In order to consider the relaxed problem, we make the following assumption.

%First, we assume that
\begin{assumption}
$U(\cdot)$ is a concave strictly-increasing function.
\end{assumption}
The intuition behind the above assumption is that users' rate of increase of the preference decreases as the play-back rate of the chunk increases.  This is a standard assumption as the user utilities are often assumed to be concave (e.g., electricity market \cite{low2}, quality of service for multimedia\cite{wang}). $U(x)=x^{\alpha}$ where $\alpha\leq 1$ is an example of a concave function.

%Hence, the problem $\mathcal{P}$ can be represented as
%\begin{eqnarray}
%\mathcal{P1}: \text{maximize } & \sum_{k=1}^{K}\mathbbm{E}[\min_{j\in V_k}U(R_{j,k})]-\lambda\sum_{k=1}^{K}%(\dfrac{d_k(\sum_{i=1}^{N_k}R_{i,k})}{C_k}-B_k)^{+}\nonumber\\
%\text{subject to }&  (\ref{b_k}),\quad B_0=T_0.\nonumber\\
%\text{var}: & R_{j,k}\in \{R_0,\ldots, R_{max}\}\nonumber
%\end{eqnarray}

Relaxing  the discrete strategy space, then the relaxed problem is given as
\begin{eqnarray}
\mathcal{P}_{\text{rel}}: \text{maximize } & (\ref{eq:pe})\nonumber\\\
\text{subject to }&  (\ref{eq:equi}), (\ref{eq:tk}), (\ref{eq:max_buffer}), \quad  R_0\le R_{i,k}\le R_{m} \nonumber\\
\text{var}: & R_{i,k}\nonumber
\end{eqnarray}
In the next proposition, we show that the above optimization problem is convex.
\begin{proposition}
The optimization problem $\mathcal{P}_{\text{rel}}$ is a convex optimization problem.
\end{proposition}
{\em Outline of Proof}:
The objective function is concave as the point-wise minimum of concave functions is concave\cite{boyd}. The strategy space is convex because of the relaxation. Hence, the problem $\mathcal{P}_{\text{rel}}$ indeed a convex optimization problem.\qed

Since the relaxed problem $\mathcal{P}_{\text{rel}}$ is convex, we can solve it using standard convex optimization techniques efficiently. However, we have to transform the optimal solution of $\mathcal{P}_{\text{rel}}$ into a feasible solution as the optimal solution $R_{i,k}$ of $\mathcal{P}_{\text{rel}}$ may not belong to the discrete set $\{R_{0},\ldots, R_{m}\}$. In the following section, we provide  heuristic solutions to obtain the feasible solutions from the optimal solution of the relaxed problem.

\section{Heuristics based on the solution of the relaxed problem and Online Algorithm}\label{sec:heuristic}
In the following, we provide a heuristic solution to obtain the feasible solution of the original problem from the relaxed one as described in $\mathcal{P}_{\mathrm{rel}}$.  The heuristic solution will be further extended to an online algorithm.
%Based on the above algorithm, we also propose an online algorithm.
\subsection{Feasible Solution}
Let $R^{*}_{i,k}$ be an optimal solution of the relaxed problem $\mathcal{P}_{\mathrm{rel}}$.
Note that $R^{*}_{i,k}$ can be efficiently computed as $\mathcal{P}_{\mathrm{rel}}$ is convex. However, we have  to discretize the solution to $\{R_{0},\ldots, R_{m}\}$.
%We first provide a simple heuristic, which finds the highest possible rate in $\{R_0,\ldots, R_{m}\}$ that is at most the rate given in $R^{*}_{i,k}$,  to find the solution.
\begin{enumerate}
\item For $k=1,\ldots, K$ do the following
\item $R_{i,k}=R_j$, where $j = \max \{u: R_{u}\leq R^{*}_{i,k}\}$.
%\item Define $L_{0}=0$ and  $L_k=L_{k-1}+\sum_{j}(R^{*}_{j,k}-R_{i,k})$.
\end{enumerate}
The above algorithm does a simple down quantization, {\it i.e.}, it fetches a tile at the highest possible rate in the set $\{R_0,\ldots, R_m\}$which does not exceed the rate given by the relaxed problem.  It is easy to discern that the solution is feasible, as described in the following Theorem.
%$L_{K}$ will provide the total rate that we have saved by down-quantization for the entire video. It is easy to discern that
\begin{theorem}\label{thm:heuristic1}
$R_{i,k}$ is a feasible solution of $\mathcal{P}_{e}$. The algorithm also does not increase the stall time as compared to the solution of the relaxed problem $\mathcal{P}_{\mathrm{rel}}$.
\end{theorem}
Note that if the quantization level are very close, then the feasible solution obtained will be very close to the optimal solution. However, if it is not the case, the solution may be  bad compared to the optimal solution as the bandwidth may be wasted. 
%\feng{Can we say something or justify that the latter case is more likely to happen in practice? Do we have any numerical results to support that?}\

In the following, we provide a better heuristic which utilizes part of the wasted bandwidths. In order to describe the proposed solution, we first introduce a few notations.

%to obtain the optimal solution.

\begin{definition}
Let $\mathcal{N}_{q,k}$ be the set of all $q$ tiles of chunk $k$.
\end{definition}
\begin{definition}
Let $\bar{A}_{q,k}$ denote the set of $q$ tiles of chunk $k$ such that $\bar{A}_{q,k}$ has the maximum probability to be part of the FoV among the $q$ tiles. Thus, mathematically,
\begin{align}
\bar{A}_{q,k}=\{S^{*}_{q,k}:S^{*}_{q,k}=\arg\max_{S_{q,k}\in \mathcal{N}_{q,k}}\Pr(FoV\subset S_{q,k})\}.
\end{align}
\end{definition}
Recall that FoV is defined as the set of $l$ tiles within a chunk. Thus, for any $q<l$, $\Pr(FoV\subset S_{q,k})=0$ and thus $\arg\max$ can be any set. $\bar{A}_{q,k}$ denotes the set which has the highest probability to be the part of the FoV among all the sets consisting of $q$ tiles. 

%set of tiles of less than $l$ can never be FoV, hence, has a zero probability.

Using the definition of $\bar{A}_{q,k}$ we can find the feasible solution from the relaxed one, which is given in Algorithm~\ref{algo1}.

%{\bf Arnob, change $R_{\max}$ to $R_m$ with $m+1$ levels }
%\begin{algorithm}
\begin{figure}
\begin{algorithm}[H]
\small
\begin{algorithmic}
\STATE {\bf Initialization:} $L_{0}=0$.
\FOR {$k=1,\ldots, K$}
\FOR {$i=1,\ldots, N$}
\STATE Set $R_{i,k}=\max \{R_{j}: R_{j}\leq R^{*}_{i,k}\}$.
\STATE Set $\bar{R}_{i,k}=\min \{R_{j}: R_{j}\geq R^{*}_{i,k}\}$.
\ENDFOR
\STATE Compute $L_k=L_{k-1}+(R^{*}_{i,k}-R_{i,k})$.
\ENDFOR
 \IF {$L_{K}=0$} 
 \STATE exit, it is the optimal solution.
 \ENDIF
\FOR {$k=1,\ldots, K$}
\STATE {\bf Initialization} $q=l$.
\IF {$L_{k}\leq \sum_{i\in \bar{A}_{q,k}}(\bar{R}_{i,k}-R_{i,k})$}
\IF {$q>l$}
\STATE Update $R_{i,k}=\bar{R}_{i,k}$ for $i\in \bar{A}_{q-1,k}$.
\STATE $L_{k+1}=L_{k}+\sum_{i\in \bar{A}_{q-1,k}}(R_{i,k}-\bar{R}_{i,k})$.
\ENDIF%\feng{($\bar{R}_{i,k}-R_{i,k}$)?}.
\ELSE
\STATE $q=q+1$.
\ENDIF
\ENDFOR
\end{algorithmic}
\caption{provides a feasible solution of $\mathcal{P}_{e}$ from the optimal solution of the relaxed version $\mathcal{P}_{\mathrm{relaxed}}$.}\label{algo1}
\end{algorithm}
\vspace{-0.3in}
\end{figure}

Algorithm~\ref{algo1} tries to increase the rate in the most viewed subset of tiles as much as possible based on the extra bandwidth that is wasted due to a lower thresholding of the rates. The next theorem states that such an algorithm will yield a feasible solution of the original problem. {\em Since the fetched rates do not decrease, the minimum and the average rates attained by the algorithm are at least the same as those achieved by the algorithm stated in Theorem~\ref{thm:heuristic1}.}

\begin{theorem}
Algorithm 1 yields a feasible solution $R_{i,k}$ for all $i$ and $k$. The value attained in $\mathcal{P}_{e}$ \ by Algorithm 1 is at least equal to the value achieved by the algorithm stated in Theorem~\ref{thm:heuristic1}. The stall time also does not increase with the above approach compared to the optimal solution of the relaxed problem $\mathcal{P}_{\mathrm{relaxed}}$.
\end{theorem}

%\feng{Need to explain Algorithm 1. It took me quite a while to roughly understand it. Also I think we also need an inner loop for $q$?}
Algorithm~\ref{algo1}  only fetches a higher quality tiles using the residual bandwidth only if does not increase the stall time, {\it i.e.}, $L_{k}$ is positive for each chunk. Hence, it does not increase the stall time compared to the  optimal solution of the relaxed problem.

 %Note that the algorithm only fetches higher quality tiles for chunk $k$ if $L_{k}$ is positive, {\it i.e.}, the bandwidth usage will be lower than the optimal solution of the relaxed problem. 
Note that if there is a residual bandwidth after down-quantization Algorithm~\ref{algo1} tries to fetch  higher quality tiles (setting $R_{i,k}$ to $\bar{R}_{i,k}$)  which will have higher probability to be a part of the FoV for the current chunk.  It then tries to fetch higher quality tiles of the most viewed set of the next chunk and so on. In the online algorithm, we adopt Algorithm~\ref{algo1} for fetching the tiles of next $W$ chunks. Since the convex problem can be efficiently solved using a low complexity algorithm our approach also has very low complexity. 

Note that we only fetch one level higher for the tiles if there is a residual bandwidth for a chunk. Even if upgrading all the tiles, there is a residual bandwidth we use it to fetch the tiles of the next chunk  instead of fetching the tiles of the current chunk at higher levels. This will make sure that there will be consistency of the qualities of the rates across the chunks. %Subsequently, it fetches the higher quality tiles whic The algorithm also fetches higher quality tiles of the chunks sequentially. For example, it fetches the higher quality  tiles for chunk $1$ if possible and then proceed to the next chunk. T However \feng{Meanwhile?}, the solution of the convex problem also makes sure that the chunks in the future will also be downloaded at some good quality depending on the estimated bandwidth. %\feng{Also comment on its low complexity.}

%{\bf Arnob, The para before Thm also explains algo in two lines. Try to merge the two paragraphs to avoid repetition.}

%The other advantage of the above algorithm also does not result into increase in the stall time.  Thus, the re-buffering time will not change compared to the relaxed solution.

\subsection{Online Algorithm}
Till now, we provided an offline algorithm. Now, we describe how to obtain an online algorithm based on the offline algorithm. We consider a {\em sliding window} type algorithm. The bandwidth and head movement is predicted for $W$ chunks ahead and then the optimal algorithm is employed for obtaining the optimal download rate.  If the $c$th chunk is being downloaded, the  online algorithm gives the download rates for the tiles in the chunks $c+1, \ldots,c+W$ by solving the following optimization problem: %For first $\beta$ number of chunks, the algorithm fetches the tiles at the base layer in order to build up the buffer.
\begin{eqnarray}
\mathcal{P}_{\text{online}}: \text{max } & \sum_{k=c+1}^{c+W}\mathbbm{E}[\min_{i\in V_{f,k}}U(R_{i,k})+\nonumber\\& \gamma \sum_{i\in V_{f,k}}U(R_{i,k})]-\lambda(\tilde{t}_{W+c}-(W+c-1)L)\nonumber\\
\text{subject to }&  (\ref{eq:equi}), (\ref{eq:tk}), (\ref{eq:max_buffer})\nonumber\\
\text{var}: & R_{i,k}\in \{R_0,\ldots,R_{m}\}\nonumber
\end{eqnarray}
The $c+1$th chunk can only start downloading at time $t_{c+1}$ which is the time when the $c$-th chunk finishes its downloading. This is a sliding window protocol so we optimize after each chunk has been downloaded with new prediction. Note that the duration of chunk is small (in the order of few seconds), hence the optimization problem $\mathcal{P}_{\text{online}}$ needs to be solved very frequently. Note that Algorithm~\ref{algo1} can be solved repeatedly because of it is of low complexity.%Hence, it will take at least $L+W_k$ time to play the $k$-th chunk. The expectation is taken over the new distribution depending on the head movements.
%\feng{Say something that this can indeed be done due to our low complexity...}

It is easy to discern that the Algorithm 1 can be easily adapted to the online version with $W+c$ in place of $K$ by solving the relaxed problem by relaxing the strategy space.

\textbf{Bandwidth Prediction Error}: Note that in the online algorithm, we assume that the bandwidth can be estimated for $W$ chunks ahead. However, if the bandwidth estimation is erroneous, then we have to account for that error. We download all the tiles of the first $\eta$ chunks at the base layer $R_0$. This will help to build the buffer and then we can apply our heuristic algorithm.% when the buffer has at least $\eta$ chunks to obtain the optimal download rates of the tiles in a chunk.

Note that once the chunk $k$ is being downloaded,  we do not change its rate it in the online approach. However, the bandwidth may be very bad while downloading the chunk $k$ which may result into the stall. We can minimize it by getting all the tiles only in the base layers\footnote{If we are using the SVC, then we can ignore the enhancement layers and stop downloading all the enhancement layers.}. On the other hand, if the bandwidth is high compared to the estimated value, we do not increase the download rate, rather we keep downloading the $k$th chunk. Running algorithm to determine fetching for the chunks after $k$ can lead to an increase in the download rates of the future chunks.

%We then download the tiles of future chunks at a higher rate.

\textbf{FoV prediction error}: Note that we take into account of the distribution of the FoV of each chunk. We can also update the estimation depending on the user's head movement for $W$ chunks ahead. If the sliding window $W$ is small, then the head movement can be predicted fairly accurately. However, if there is an error, the user can still see it in the base layer, {\it i.e.}, the user will not see any black spot. {\em This is because we specify that each tile must have to be fetched at least in the base layer.} Note that with the current head position, we can again estimate the head movement for $W$ chunks ahead and can obtain optimal solution.

%In the next section, we consider a robust optimization problem where with a low probability the user will watch lower quality video%.{\bf Arnob, sentence is complex - are you saying minimizing probability of watching low quality?} 
\section{Guaranteed Rate with a given probability}\label{sec:robust}% \feng{The title sounds vague. Change it to something more meaningful.}}
Till now, we considered that the video provider wants to maximize the expected QoE. However, the above approach does not provide any probabilistic guarantee on the rate a user will watch the video. For example, Algorithm~\ref{algo1} may maximize the expected QoE, however, still the user can watch the video in a poor quality with a high probability if the variance is high.  Thus, the video player provider may want to provide a probabilistic bound for the rate that a user will watch the video. We now consider such a QoE metric.

\subsection{Problem Formulation}
First, we introduce a notation which we use throughout this section.
\begin{definition}
Let $\mathcal{A}_{\alpha,k}$ be the set of tiles of chunk $k$ which have the probability that the FoV is a subset of $\mathcal{A}_{\alpha,k}$ with probability at least $\alpha$ for chunk $k$. If there are multiple sets which satisfy the above condition, then, $\mathcal{A}_{\alpha,k}$ is considered to be the set with the lowest cardinality.\footnote{Cardinality of the set is the number of elements in the set.}
\end{definition}
Note that  the video content provider can obtain $\mathcal{A}_{\alpha,k}$ by estimating the FoV of a user and from crowd-sourced viewing statistics.

$\mathcal{A}_{\alpha,k}$ is the set of the tiles of the lowest cardinality which specifies that the FoV will be a subset of this set with probability $\alpha$. Thus, the lowest rate among the tiles of $\mathcal{A}_{\alpha,k}$ gives the lowest rate the user will observe the video with the probability $\alpha$. The video content provider wants to provide a guarantee of the rate a user will experience with probability $\alpha$. Hence, the QoE is governed by the minimum rate among those tiles within the set $\mathcal{A}_{\alpha,k}$. QoE also decreases as the stall time increases, hence, 
\begin{align}\label{eq:probust}
\text{QoE}=\sum_{k=1}^{K}U(\min_{i\in \mathcal{A}_{\alpha,k}}R_{i,k})-\lambda(\tilde{t}_K-(K-1)L-t_{\mathrm{ini}}).
\end{align}
$\lambda$ is the weight factor corresponding to the stall time (cf.(\ref{stall})). The video provider wants to maximize the above. %utility is the minimum among the rates of the tiles within the set $\mathcal{A}_{\alpha,k}$. %The minimum rate of the tiles within $\mathcal{A}_{\alpha}$ specifies that the video will be seen at least with that rate. The video service provider's objective is now governed with the minimum rate of among the tiles within the set $\mathcal{A}_{\alpha,k}$. If the minimum rate is poor, then, there is a possibility that the quality of service is poor. On the other hand, if the rate is high, the quality of service is also high at least with probability $\alpha.$
Thus, formally, the optimization problem is
\begin{eqnarray}
\text {maximize } &\sum_{k=1}^{K}U(\gamma_k)-\lambda(\tilde{t}_K-(K-1)L-t_{\mathrm{ini}})\label{eq:robust}\\
\text{subject to } &  (\ref{eq:equi}), (\ref{eq:tk}), (\ref{eq:max_buffer})\nonumber\\
& \gamma_k= \min_{i\in \mathcal{A}_{\alpha,k}}R_{i,k} \label{eq:constraint}\\
 \text{var}: & R_{i,k}\in \{R_0,\ldots, R_{m}\}\nonumber
\end{eqnarray}
Note from (\ref{eq:constraint}) that $\gamma_k$ denotes the minimum rate of the tiles within the set $\mathcal{A}_{\alpha,k}$.

Next, we provide an equivalent representation of the optimization problem (\ref{eq:robust}) which will help us to obtain a relaxed problem which is convex.
\begin{proposition}
Consider the following optimization problem
\begin{eqnarray}
& \mathcal{P}_{\mathrm{robust}}: \text {maximize } & (\ref{eq:robust})\nonumber\\
& \text{subject to } &  (\ref{eq:equi}), (\ref{eq:tk}), (\ref{eq:max_buffer})\nonumber\\
& & \gamma_k\leq R_{i,k} \forall i\in \mathcal{A}_{\alpha,k}\label{eq:robust_constraint}\\
& \text{var}: & R_{i,k}\in \{R_0,\ldots, R_m\}, \gamma_k\nonumber
\end{eqnarray}
The optimization problem $\mathcal{P}_{\mathrm{robust}}$ is equivalent to (\ref{eq:robust}).
\end{proposition}
The above result shows that even though the constraint in (\ref{eq:robust_constraint}) is not equal to (\ref{eq:constraint}), yet in any optimal solution $\gamma_k$ must satisfy (\ref{eq:constraint}).

{\em Outline of the Proof}: $U(\cdot)$ is strictly increasing function. Hence, at the optimal value we must have $\gamma_k$ satisfy constraint (\ref{eq:constraint}). \qed

\begin{remark}
$\alpha$ is a parameter which needs to be determined from the market research. In stochastic optimization problems, the robustness is an important issue where the optimizer wants to protect itself from the error in data or high variance. However, too high $\alpha$ may give very pessimistic solution. For robust optimization problems, $\alpha$ is generally taken to be in the interval $[0.9, 0.99]$. %\feng{How do you come up with this magic number of 0.9?}
\end{remark}
\begin{remark}
The optimization problem $\mathcal{P}_{\mathrm{robust}}$ is not convex because the decision space is discrete. In the following we propose a relaxed version of the problem which is convex.
\end{remark}
\begin{remark}
Note that in an optimal solution, the algorithm will never try to fetch any tile which is not in $\mathcal{A}_{\alpha,k}$ at a higher quality compared to $R_{0}$ unless it fetches all the tiles within $\mathcal{A}_{\alpha,k}$. This is because the utility will not increase if the tiles which is not in $\mathcal{A}_{\alpha,k}$.
\end{remark}
\subsection{Relaxation}
The problem in $\mathcal{P}_{\mathrm{robust}}$ is not a convex problem. We consider a relaxed version of the problem
\begin{eqnarray}
\mathcal{P}^{\mathrm{relaxed}}_{\mathrm{robust}}: & (\ref{eq:robust})\nonumber\\%\text {maximize } &\sum_{k=1}^{K}U(\gamma_k)-\lambda(\tilde{t}_K-(K-1)L-t_{\mathrm{ini}})\nonumber\\
& \text{subject to } &  (\ref{eq:equi}), (\ref{eq:tk}), (\ref{eq:max_buffer}), (\ref{eq:robust_constraint})\nonumber\\
& & R_{0}\leq R_{i,k}\leq R_m\nonumber\\
%& & \gamma_k\leq R_{j,k} \forall j\in \mathcal{A}_{\alpha}\label{eq:robust_constraint}\\
& \text{var}: & \gamma_k, R_{i,k}\nonumber
\end{eqnarray}

\begin{proposition}\label{prop:convexrobust}
$\mathcal{P}^{\mathrm{relaxed}}_{\mathrm{robust}}$  is a convex optimization problem.
\end{proposition}
Since the relaxed problem $\mathcal{P}^{\mathrm{relaxed}}_{\mathrm{robust}}$ is convex, we can solve it using the standard convex optimization techniques efficiently. However, we have to transform the optimal solution of $\mathcal{P}^{\mathrm{relaxed}}_{\mathrm{robust}}$ into a feasible solution as the optimal solution $R^{*}_{i,k}$ of $\mathcal{P}^{\mathrm{relaxed}}_{\mathrm{robust}}$  may not belong to the discrete set $\{R_0, . . . , R_m\}$. In the following section, we provide heuristic solutions to obtain the feasible solutions from the optimal solution of the relaxed problem.%Thus, $\mathcal{P}^{\mathrm{relaxed}}_{\mathrm{robust}}$ can be efficiently solved using existing numerical approaches.

\subsection{Heuristic to obtain the feasible solution}
Let $R^{*}_{i,k}$ be the optimal solution of the relaxed problem $\mathcal{P}^{\mathrm{relaxed}}_{\mathrm{robust}}$ for all $i$ and $k$. $R_{i,k}^{*}$ can be obtained very fast as the problem is convex (Proposition~\ref{prop:convexrobust}). Now, we describe a simple heuristic to obtain a feasible solution based on the optimal solution.
%\begin{algorithm}
\begin{figure}
\begin{algorithm}[H]
\small
\begin{algorithmic}
\STATE {\bf Initialization:} $L_0=0$.
\FOR{$k=1,\ldots, K$}
\FOR{$i=1,\ldots, N$}
\STATE $R_{i,k}=\max \{R_{j}: R_{j}\leq R^{*}_{i,k}\}$.
\STATE $\bar{R}_{i,k}=\min \{R_{j}: R_{j}\geq R^{*}_{i,k}\}$.
\ENDFOR
\STATE Compute $L_k=L_{k-1}+(R^{*}_{i,k}-R_{i,k})$.
 \STATE $\gamma_k=\min\{R_{i,k}:i\in \mathcal{A}_{\alpha,k}\}$.
\ENDFOR
\IF{$L_{K}=0$} 
\STATE Exit since it is the optimal solution.
\ENDIF
\FOR{$k=1,\ldots,K$}
\IF{$L_k>\sum_{i\in \mathcal{A}_{\alpha,k}}(\bar{R}_{i,k}-R_{i,k})$}
\STATE Update: $R_{i,k}=\bar{R}_{i,k}$ for all $i\in \mathcal{A}_{\alpha,k}$.
\STATE $L_{k+1}=L_{k}-\sum_{i\in \mathcal{A}_{\alpha,k}}(\bar{R}_{i,k}-R_{i,k})$.
\STATE Update $\gamma_k=\min\{R_{i,k}:i\in \mathcal{A}_{\alpha,k}\}$.
\ENDIF
\ENDFOR
\end{algorithmic}
\caption{provides a feasible solution of $\mathcal{P}_{\mathrm{robust}}$ from the optimal solution of $\mathcal{P}^{\mathrm{relaxed}}_{\mathrm{robust}}$}\label{algo2}
\end{algorithm}
\vspace{-0.3in}
\end{figure}
%\begin{algorithmic}
% $L_0=0$\;
%\For{$k=1,\ldots, K$}{
%\For{$i=1,\ldots, N$}{
% $R_{i,k}=\max \{R_{j}: R_{j}\leq R^{*}_{i,k}\}$\;
%$\bar{R}_{i,k}=\min \{R_{j}: R_{j}\geq R^{*}_{i,k}\}$\;
%}
%Compute $L_k=L_{k-1}+(R^{*}_{i,k}-R_{i,k})$\;
%}
%\If{$L_{K}=0$} {exit, it is the optimal solution\;}
%\For{$k=1,\ldots,K$}{
%\If{$L_k-\sum_{i\in \mathcal{A}_{\alpha,k}}(\bar{R}_{i,k}-R_{i,k})>0$}{
%Update: $R_{i,k}=\bar{R}_{i,k}$ for all $i\in \mathcal{A}_{\alpha,k}$\;
%$L_{k+1}=L_{k}-\sum_{i\in \mathcal{A}_{\alpha,k}}(\bar{R}_{i,k}-R_{i,k})$\;}
%}
%%\end{algorithmic}
%\caption{Algorithm 2}\label{algo2}
%\end{algorithm}
%\textbf{Algorithm 2}:
%\begin{enumerate}
%\item $L_0=0$.
%\item For $k=1,\ldots, K$ and $i=1,\ldots, N$ do the following
%\item Set $R_{i,k}=\{R_{j}: R_{j}\leq R^{*}_{j,k}<R_{j+1}\}$,
%\item Set $\bar{R}_{i,k}=\{\min\{R_j,R_{max}\}:R_j>R_{i,k}, R_{j-1}=R_{i,k}\}$.
%\item Compute $L_k=L_{k-1}+\sum_{i}(R^{*}_{i,k}-R_{i,k})$.
%\item End For.
%\item If $L_{K}>0$, then go to the next step, otherwise exit since it is the optimal solution.
%\item For $k=1,\ldots, K$ do the following
%\begin{enumerate}
%\item If $L_{k}-\sum_{i\in \mathcal{A}_{\alpha,k}}(\bar{R}_{i,k}-R_{i,k})> 0$,
%\item Update: $R_{i,k}=\bar{R}_{i,k}$ for $i\in \mathcal{A}_{\alpha}$ and $L_{k+1}=L_{k}-\sum_{i\in \mathcal{A}%_{\alpha,k}}(\bar{R}_{i,k}-R_{i,k})$.
%\end{enumerate}
%\item End For.
%\end{enumerate}
Note that the Algorithm~\ref{algo2} first computes $L_{k}$, the bandwidth that is saved because of the down-quantization of the rates. Then the algorithm tries to fetch the tiles within the set $\mathcal{A}_{\alpha,k}$ at one level higher (assigning to $\bar{R}_{i,k}$)  provided that the stall time does not increase compared to the optimal solution of the relaxed problem ($L_{k}$ is not negative).  Thus, similar to Algorithm 1, the algorithm tries to minimize any wasted bandwidth because of the down-quantization. Hence, it will achieve higher rate compared to the simple down-quantization. Note that the algorithm tries to fetch the higher quality tiles for chunk $k$, then proceeds to fetch higher quality ones for chunk $k+1$ and so on.  %The FoV can be more accurately predicted for chunks in the near future. The chunks in the near future have tighter deadlines, hence, we may see low quality chunks at the beginning, but very high quality ones in the future. %Thus. there mayThus, there may be smaller tiles in the set of $A_{\alpha,k}$ for the chunk $k$ in the near future compared to the distant future. Hence, smaller amount of bandwidth may be required for the chunks of near future.

\begin{theorem}
The above algorithm gives a feasible solution of the problem $\mathcal{P}_{\mathrm{robust}}$. The stall time does not increase compared to the optimal solution of the relaxed problem.
\end{theorem}
Note that if the quantization levels are very close, then the feasible solution obtained will be very close to the optimal solution.%If the quantization levels are minute, the algorithm closely matches the optimal solution. Also note that the stall time remains the same as in the relaxed problem.

\subsection{Optimal Algorithm for a class of utility function}\label{sec:optimal}
In the previous section we describe a simple heuristic algorithm which gives a feasible solution from the relaxed problem. The algorithm performs good if there is a continuum of the achievable rates. However, in practice this may not happen. In the following, we describe an algorithm (Algorithm~\ref{algo3})  for selecting the download rates for tiles in a chunk and we show that the algorithm is optimal even for the original problem $\mathcal{P}_{\mathrm{robust}}$ for a class of utility functions under some assumptions which frequently arise in practice.
\begin{figure}
\begin{algorithm}[H]
\small
\begin{algorithmic}
\STATE {\bf Initialization:} $R_{i,k}=R_0$ for all $i$ and $k$.
\STATE $R^{\prime}_{i,k}=R_{i,k}$ for all $i$ and $k$.
\FOR{$k=1,\ldots, K$}
\STATE $\gamma_k=\min\{R_{i,k}:i\in \mathcal{A}_{\alpha,k}\}$.
\STATE Compute $t_k$ according to (\ref{eq:tk}).
\STATE $t^{\prime}_k=t_k$.
\ENDFOR
\IF{$t_{K}\geq (K-1)L+t_{\mathrm{ini}}$}
\STATE Exit.
\ENDIF
\FOR{$k=1,\ldots, K$}
\FOR{$j=1,\ldots, m$}
\STATE Set $R^{\prime}_{i,k}=R_j$ for all $i\in \mathcal{A}_{\alpha,k}$.
\FOR{$f=k,\ldots, K$}
\STATE Compute $t^{\prime}_{f}=t^{\prime}_{f-1}+L\dfrac{\sum_{i}R^{\prime}_{i,f}}{C_f}$.
\ENDFOR
\IF{$t_K\leq (K-1)L+t_{\mathrm{ini}}$}
\STATE Set $R_{i,k}=R^{\prime}_{i,k}$ for all $i$.
\STATE Update $t_{k}=t^{\prime}_{k}$.
\ENDIF
\ENDFOR
\STATE Update $\gamma_k=\min\{R_{i,k}:i\in \mathcal{A}_{\alpha,k}\}$.
\ENDFOR
\end{algorithmic}
\caption{gives an optimal solution of $\mathcal{P}_{\mathrm{robust}}$ for a class of utility functions under some assumptions stated in Theorem~\ref{thm:optimal}}\label{algo3}
\end{algorithm}
\vspace{-0.3in}
\end{figure}
%\begin{algorithm}
%$R_{i,k}=R_0$ for all $i$ and $k$\;
%$R^{\prime}_{i,k}=R_{i,k}$ for all $i$ and $k$\;
%\For{$k=1,\ldots, K$}{
%Compute $t_k$ according to (\ref{eq:tk})\;
%}
%\If{$t_{K}\geq (K-1)L+t_{\mathrm{ini}}$}{Exit\;}
%\For{$k=1,\ldots, K$}{
%\For{$j=1,\ldots, m$}{
%Set $R^{\prime}_{i,k}=R_j$ for all $i\in \mathcal{A}_{\alpha,k}$\;
%\For{$a=k,\ldots, K$}{
%Compute $t_{a}=t_{a-1}+L\dfrac{\sum_{i}R^{\prime}_{i,a}}{C_a}$\;
%}
%\If{$t_K\leq (K-1)L+t_{\mathrm{ini}}$}{
%Set $R_{i,k}=R^{\prime}_{i,k}$ for all $i$\;
%}
%}
%}
%\caption{Algorithm 3}\label{algo3}
%\end{algorithm}
%\textbf{Algorithm 3}:
%\begin{enumerate}
%\item Initialization: $R_{i,k}=R_0$ for all $i$ and $k$.
%\item Set $R^{\prime}_{i,k}=R_{i,k}$ for all $i$ and $k$.
%\item For $k=1,\ldots, K$
%\item Compute $t_k$ according to (\ref{eq:tk}).
%\item End For.
%\item If $t_{K}\geq (K-1)L+t_{\mathrm{ini}}$, then exit.
%\item Else: For $k=1,\ldots,, K$ do the following
%\begin{enumerate}
%\item For $j=1,\ldots, m$
%\item Set $R^{\prime}_{i,k}=R_j$ for all $i\in \mathcal{A}_{\alpha,k}$.
%\item For $l=k,\ldots, K$
%\item Compute $t_l=t_{l-1}+L\dfrac{\sum_{i}R^{\prime}_{i,l}}{C_l}$.
%\item End For.
%\item If $t_K\leq (K-1)L+t_{\mathrm{ini}}$,
%\item Then: Set $R_{i,k}=R^{\prime}_{i,k}$.
%\item Else : Exit the inner loop.
%\item End For.
%\end{enumerate}
%\item End For.
%%\item Compute $\tilde{t}_1, \tilde{t}_2$Rank in the non-increasing order according to $\dfrac{U(R_{i})}{\sum_{j\in \mathcal{A}_{\alpha}}R_i+R_0}$ for each $i$ and $k$.
%%\item For $k=1,\ldots, K$, we have
%%\item $R_{j,k}=R_{i}$ for $j\in \mathcal{A}_{\alpha}$.
%\end{enumerate}

Recall from (\ref{eq:tk})  that $t_k$ is the start time of downloading chunk $k$. In  Algorithm~\ref{algo3}, first all the tiles are fetched at the base rate $R_0$. Now for each chunk $k$, it fetches the tiles within the set $\mathcal{A}_{\alpha,k}$ at the highest possible rate if the download-time does not exceed the play-back time {\it i.e.} it does not increase the stall. The output of the algorithm provides the rate $R_{i,k}$ at which the tile $i$ of chunk $k$ has to be downloaded. The algorithm is very simple to implement. Note that  Algorithm 3 {\em only} fetches the tiles at a higher quality within the set $\mathcal{A}_{\alpha,k}$ for chunk $k$. It {\em does not fetch} all the tiles  within chunk $k$ at a higher rate.

The cardinality of $\mathcal{A}_{\alpha,k}$ is higher if $\alpha$ is high or the distribution of the FoV estimation has a higher variance. Thus, if there is a higher variance regarding the FoV, the algorithm needs to fetch more tiles in order to guarantee a rate with probability $\alpha$. 

Note that one drawback of the above approach is that the qualities may vary over different chunks. For example, if the future bandwidth is very poor, it will fetch the tiles at a lower qualities; though it can fetch the tiles of the current chunk at a higher rate as the current bandwidth is high. To avoid the above, one can reduce the maximum rate one can fetch the tiles for current chunks and use the residual bandwidth for downloading tiles not belonging to $\mathcal{A}_{\alpha,k}$.

In the online version, we only optimize for some $W$ chunks ahead ($W<K$). Thus, the online version can be computed very fast. %To avoid that, one approach can be to restrict the maximum possible rate for the chunks of the near future at some lower rates.

In the following we show that the above algorithm is optimal under some assumptions.
\begin{assumption}\label{assum:ordering}
$|\mathcal{A}_{\alpha,k}|\geq |\mathcal{A}_{\alpha,k-1}|$ for all $K\geq k\geq 2$.
\end{assumption}
Recall that the FoV is estimated. The above assumption entails that the estimate is assumed to be more accurate in near future rather than the distant one as we need fewer tiles to cover the FoV with probability $\alpha$. In practice, we also commonly observe that the prediction is more accurate in the near future compared to the distant future\cite{bao}. Hence, the high probability region should consist of fewer tiles in the near future compared to the distant ones. %Hence, it is more likely that few tiles are required to constitute the high visibility region for chunks in the near future.

Finally, we assume that
\begin{assumption}\label{ref:linearutility}
$U(x)=ax+b$, where $a$, $b$ are constant.
\end{assumption}
Assumption~\ref{ref:linearutility} entails that the utility varies linearly with the rate.
\begin{theorem}\label{thm:optimal}
Algorithm 3 provides an optimal solution of $\mathcal{P}_{\mathrm{robust}}$ if Assumptions~\ref{assum:ordering} and \ref{ref:linearutility} hold, and $\lambda$ is large enough\footnote{Formally, $\lambda>\sum_{k=1}^{K}U(R_{max})$}.
\end{theorem}
\begin{proof}
See Appendix~\ref{sec:proof}.
\end{proof}

The above theorem thus, tells that under reasonable assumptions a simple heuristic algorithm such as Algorithm 3 can be optimal for a large class of utility functions.

Intuitively,  Assumption~\ref{assum:ordering} entails that one can gain same utility by consuming lower bandwidth by fetching higher quality tiles  for nearby chunk as it needs to fetch smaller number of tiles compared to the chunks of distant future.  Assumption \ref{ref:linearutility} further characterizes that one will not lose any utility by fetching the best possible quality tiles of the current chunk instead saving the bandwidth for fetching a higher quality tiles for later chunks. %fetching higher quality tiles of the current chunk  compared to  tiles of higher qualities of the distant future. \feng{Why the implications for Assumptions 2 and 4, which just define the rates and the utility functions?}

\subsection{Online Algorithm}
Now, we describe how to obtain an online algorithm based on the offline algorithms~\ref{algo2} and \ref{algo3}. We consider a {\em sliding window} type algorithm. The bandwidth and head movement is predicted for $W$ chunks ahead and then the optimal algorithm is employed for obtaining the optimal download rate for each tile within a chunk.  If the tiles of the $c$-th chunk is being fetching, the online algorithm gives the download rates for chunks $c+1, \ldots,c+W$ by solving the following optimization problem:
\begin{eqnarray}
\mathcal{P}^{\mathrm{online}}_{\mathrm{robust}}: & \text {maximize } &\sum_{k=c+1}^{c+W}U(\gamma_k)\nonumber\\ & &-\lambda(\tilde{t}_{W+c}-(W+c-1)L-t_{\mathrm{ini}})\nonumber\\
& \text{subject to } &  (\ref{eq:equi}), (\ref{eq:tk}), (\ref{eq:max_buffer}), (\ref{eq:robust_constraint})\nonumber\\
%& & \gamma_k\leq R_{j,k} \forall j\in \mathcal{A}_{\alpha}\label{eq:robust_constraint}\\
& \text{var}: & \gamma_k\geq 0,  R_{i,k}\in \{R_0,\ldots, R_{m}\}\nonumber
\end{eqnarray}
The $c+1$th chunk can only start downloading at time $t_{c+1}$. Since this is a sliding window protocol so we optimize after each chunk has been downloaded with new prediction of the FoV and the bandwidth. Note that the width of chunk is small (in the order of few seconds), hence the optimization problem $\mathcal{P}^{\mathrm{online}}_{\mathrm{robust}}$ needs to be solved very frequently.   %Hence, it will take at least $L+W_k$ time to play the $k$-th chunk. The expectation is taken over the new distribution depending on the head movements.

It is easy to discern that the Algorithms ~\ref{algo2} and \ref{algo3} can be easily adapted to the online version with $W+c$ in place of $K$. Because of the low complexity nature of the Algorithms~\ref{algo2} and \ref{algo3} they can be easily adapted for solving the online version repeatedly. Algorithm 2 is obtained from solving the relaxed version of the problem by relaxing the discrete strategy space as discussed in Section V-B.

\textbf{Bandwidth Prediction Error}: We can solve the problem of bandwidth prediction error by making the criteria that the first $\eta$ tiles need to be downloaded at the base layer before using our proposed algorithms as we mentioned in Section IV-B. %Note that in the online algorithm, we assume that the bandwidth can be estimated for $W$ chunks ahead. However, if the bandwidth estimation is erroneous, then we have to account for that error. We will make all the tiles of the first $\eta$ chunks have to be downloaded at the base layer $R_0$. This will help to build the buffer and then we can apply our heuristic algorithm to obtain the optimal download rates of the tiles in a chunk.

%Note that once we decide the rates at which the chunk $i$ to be downloaded, we do not change it in the online approach. However, the bandwidth may be very bad while downloading the chunk $i$ which may result into the stall. We can minimize it by getting all the tiles only in the base layers\footnote{If we are using the SVC, then we can ignore the enhancement layers and stop downloading all the enhancement layers.} if the buffer does not have at least $\eta$ chunks ahead. On the other hand, if the bandwidth is high compared to the estimated value, we do not increase the download rate, rather we keep downloading the $i$th chunk. We then download the tiles of future chunks at a higher rate.

\textbf{FoV prediction error}: Since in the optimization problem we try to fetch the tiles which have high probability to be the part of the FoV. Hence, the user may watch the same quality video with a very high probability.   Thus, the impact of the FoV prediction error will be low compared to Algorithm 1. Similar to Algorithm 1, if the FoV consists of tiles which are not in the part of $\mathcal{A}_{\alpha,k}$, the user can still watch the video in the base layer. 
\section{Numerical Evaluation}\label{sec:simulation}
We, now, evaluate the strengths of our proposed algorithms in a simulated setting.
\subsection{Parameter Setup}
The utility function is considered to be linear. Without loss of generality, we consider $U(x)=x$. The tile configuration for each chunk is considered to be $4\times 8$ which maps the overall 360 degree video in 2D. The FoV is considered to be at-most 120 degrees in the horizontal direction and 120 degree in the vertical direction\cite{feng}. Hence, in our case the FoV consists of a rectangle of $2\times 3$ tiles. The duration of each chunk is considered to be 2 seconds. The total number of chunks ($K$) is $120$. Hence, the length of the video is $4$ minutes.

Each chunk's rate is $x$ Mbps. We consider $x\in \{8,16,24,32\}$. Hence, each tile's rate is $x/32$ Mbps; $x/32 \in \{0.25, 0.5, 0.75 ,1\}$.  

For bandwidth traces, we used a dataset from \cite{trace}, which consists of continuous 1-second measurement of video streaming throughput of a moving device. However, the trace is for HDSPA mobile network. In the current 4G-LTE network, the rate is typically 5-10 times higher. In order to fit in the LTE setting, we linearly scale the BW profile by 5 times.   We use 40 traces which are at least 240 seconds long. In our simulation, the predicted bandwidth is computed by multiplying the actual value in the bandwidth trace by $1+e$ where $e$ is uniformly drawn from $[-p,p]$ where $0\leq p< 1$. Higher $p$ indicates higher variability. %We study the impact o$p$.

We employ the online algorithm. We run the online algorithm after the completion of the download of the each chunk for $W$ chunks ahead. Recall that $W$ is the size of the sliding window.   %Note that in all our algorithms, we optimize repeatedly after downloading each chunk. The algorithms are fast to provide results, hence, we can easily adapt after every $1$ or $2$ seconds.

 \cite{feng} shows that the FoV can be predicted with a higher accuracy in the short run. The estimation of the FoV is thus considered to be the following:  for each chunk, the FoV coincides with a  particular $2\times 3$ tiles with a high probability $\beta\geq 0.5$ (Fig.~\ref{fig:fov}). The specific area can vary over the chunks. The  event that FoV is outside this region is considered to be uniformly distributed. Hence, any other tile other than those tiles have equal probability to be in the FoV.  Note that the cardinality of $\mathcal{A}_{\alpha,k}$ is the same for all $k$. {\em Higher $\beta$ indicates that there is more certainty regarding the FoV (Fig.~\ref{fig:fov}).} For example, when $\beta=1$, the FoV is fully known. Lower $\beta$ indicates there is more uncertainty regarding the FoV. 

Recall from (\ref{eq:pe})  or (\ref{eq:probust}) that $\lambda$ is the weight corresponding to minimizing the stall time. We consider that $\lambda=1000$ {\it i.e.}, we give more preference to minimize the stall time.  Note that since the utility is linear and the $|\mathcal{A}_{\alpha,k}|$ is the same for all $k$, hence by Theorem~\ref{thm:optimal} Algorithm 3 is optimal for $\mathcal{P}_{\text{robust}}$. We, thus, evaluate  the strengths of our proposed  Algorithms 1 and 3. For Algorithm 3, we set $\alpha$ at $0.95$. For comparison, we consider two simple strategies. These are explained in the following:

\textbf{Algorithm without taking into account of the FoV prediction}: The algorithm will try to fetch all the tiles of the same quality if the bandwidth permits. In other words, it does take into account of the FoV prediction.  Thus, the qualities will be the same across the tiles, however, the quality may not be very good for viewing if the bandwidth is low.  The algorithm tries to fetch the tiles of the current chunk in the highest possible quality if the bandwidth permits before fetching the higher quality tiles of the next chunk. This is a simple algorithm, and thus, it is used widely in practice \cite{facebook}. We denote this algorithm as \textbf{Baseline algorithm}. Note that  similar to Algorithm~\ref{algo3} (Section~\ref{sec:optimal}) the Baseline  algorithm tries to fetch the higher quality tiles of the current chunk before moving to the next chunk. The only difference is that Algorithm~\ref{algo3}  only fetches the higher quality tiles which belong to the set $\mathcal{A}_{\alpha,k}$ in chunk $k$ and the rest of the tiles are fetched at the base rate whereas the Baseline algorithm fetches {\em all} the tiles in the same quality. %  \feng{the following is unclear: in your Baseline algorithm, when you have extra BW, do you fetch the next chunk's tiles at a lower layer, or the current chunk's tiles at a higher layer?}

\textbf{Greedy Algorithm}:
This algorithm is proposed by \cite{feng}. The algorithm only fetches those tiles which have the highest probability to be the part of FoV. Again similar to the {\bf Baseline} algorithm it fetches the higher quality tiles of the current chunk first if there is an extra bandwidth before it fetches the tiles of the next chunk.  Though in this algorithm only a few tiles are required to be downloaded, if the prediction is not accurate, the viewer can see black spots even though it can save a lot of bandwidth. %\feng{This greedy algorithm is also not very clear. Need to provide more details of how it performs rate adaptation.}

{\em We study the strength of our proposed algorithms Algorithms~\ref{algo1} (Section~\ref{sec:heuristic}) and \ref{algo3}(Section~\ref{sec:optimal}) with respect to these two algorithms. } We evaluate the algorithms based on the QoE metrics (Sections~\ref{sec:qos} and \ref{sec:robust}), the distribution of the bitrates in the FoV and the stall duration. 

%\feng{Remind readers what Algorithms 1 and 3 are. Readers tend to forget things.}
%
%\feng{Also need to say what are the evaluation metrics: bitrate, stall duration, (what else?)}
%
%\feng{Font sizes in Figure 2 and 3 are too small.}

%\begin{figure}[htbp]
%\begin{center}
%\epsfile{file=,scale=0.8}
%\caption{{\bf default}}
%\label{default}
%\end{center}
%\end{figure}

\subsection{Results and Discussions}
\begin{figure*}
\begin{minipage}{0.35\textwidth}
\includegraphics[width=\textwidth]{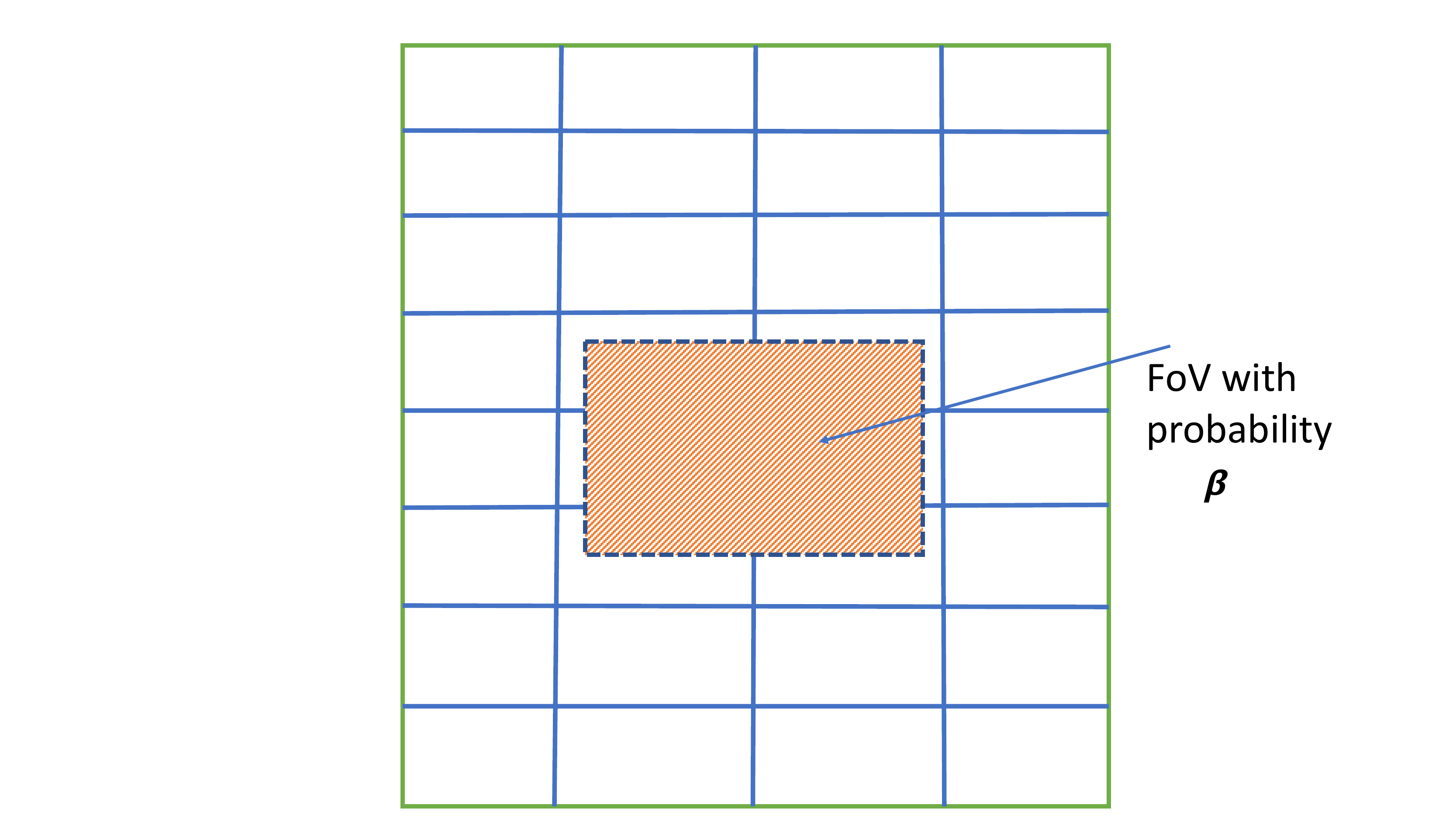}
\vspace{-0.2in}
\caption{The tile based segmentation and the rectangular $2\times 3$ tiles is FoV with probability $\beta$.}
\vspace{-0.2in}
\label{fig:fov}
\end{minipage}\hfill
\begin{minipage}{0.3\textwidth}
%\epsfile{file=gammavsoptimization,scale=0.8}
\includegraphics[width=\textwidth]{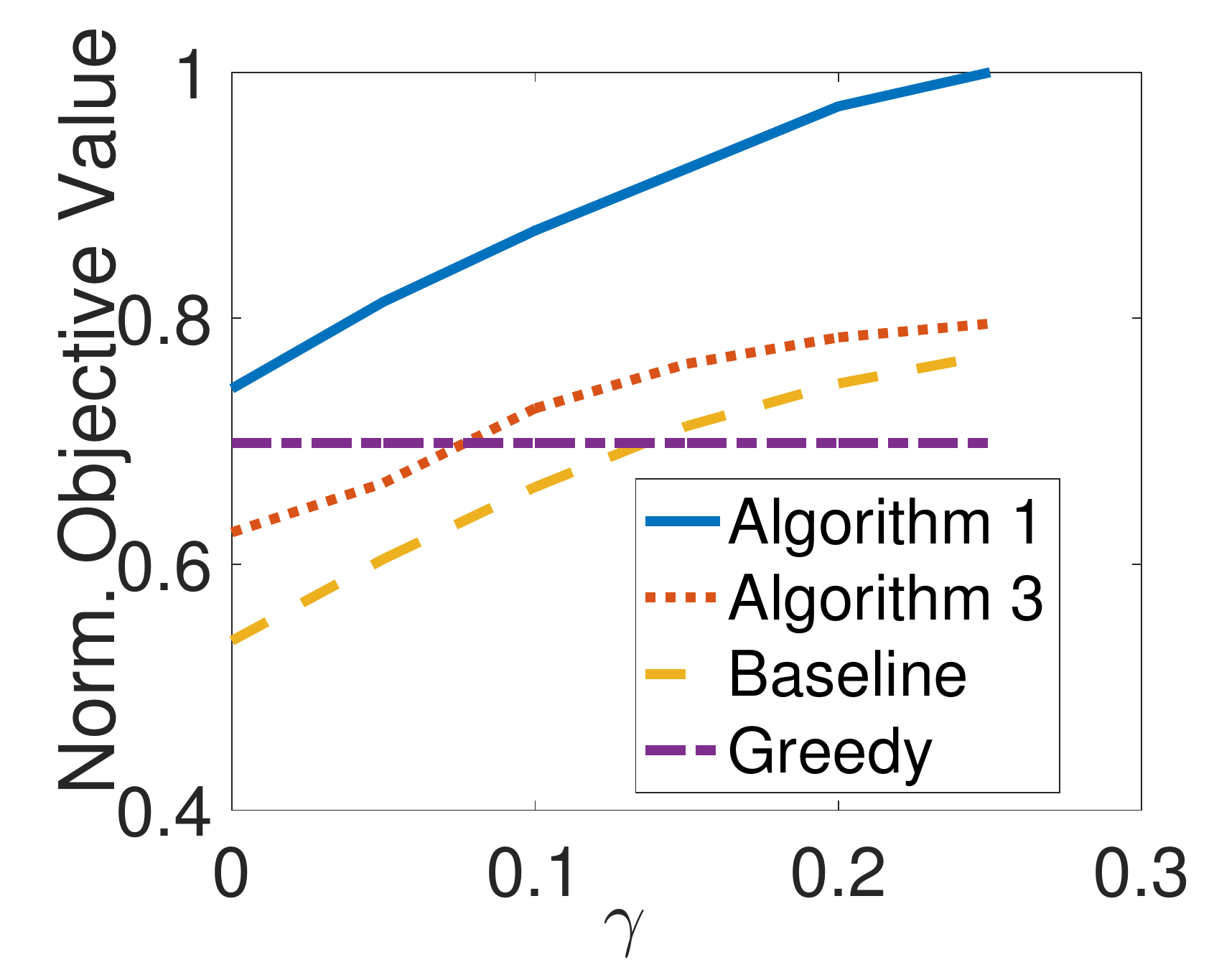}
\vspace{-0.2in}
\caption{The variation of the normalized objective function in $\mathcal{P}$ with $\gamma$ for $\beta=0.8, p=0.2$, and $W=2$.}
\label{fig:gamma}
\vspace{-0.2in}
\end{minipage}\hfill
\begin{minipage}{0.3\textwidth}
\includegraphics[width=\textwidth]{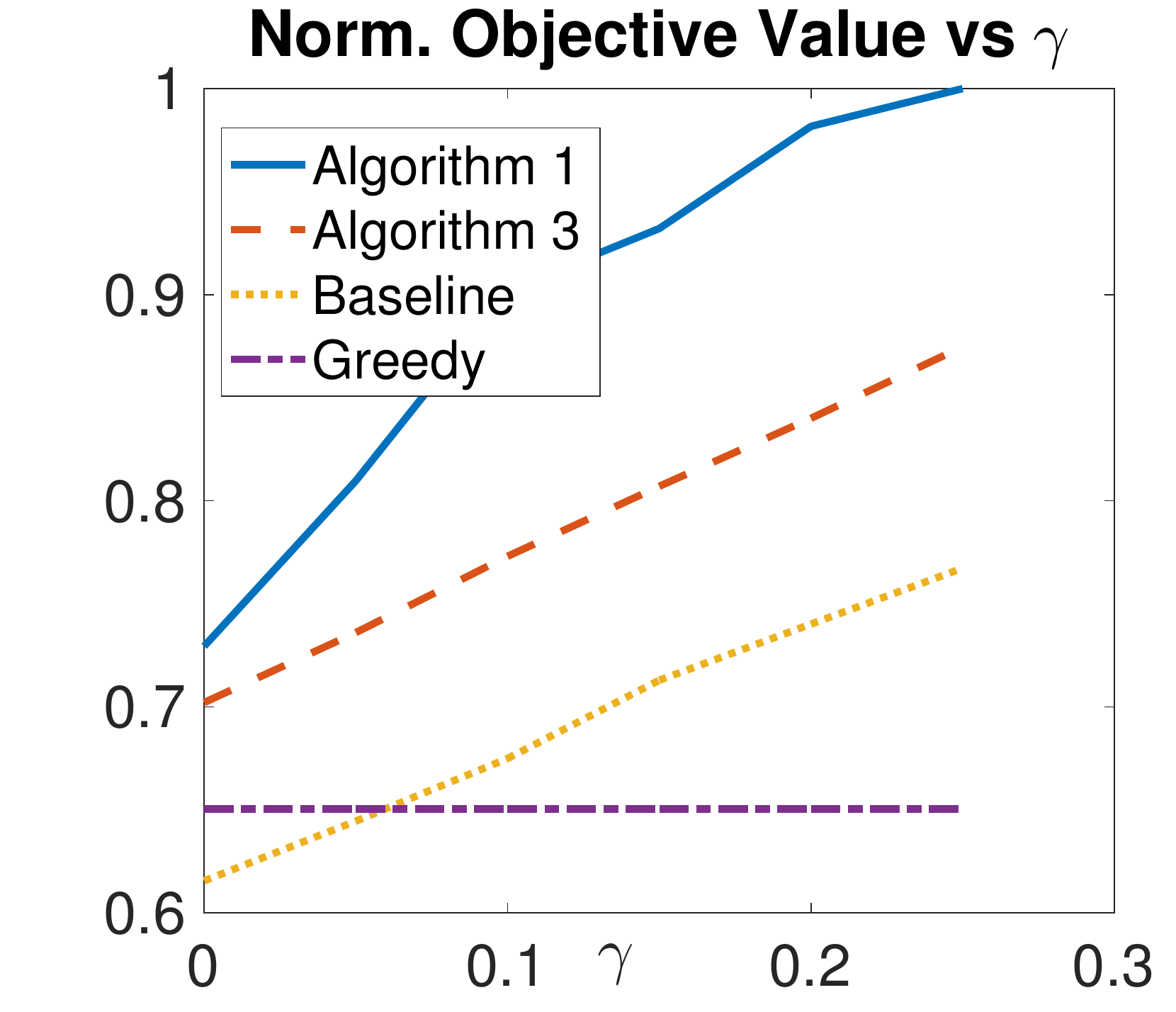}
\vspace{-0.2in}
\caption{The variation of the normal objective function in $\mathcal{P}$ with $\gamma$ for $\beta=0.6, p=0.2$ and $W=2$.}
\vspace{-0.2in}
\label{fig:gammavserror}
\end{minipage}
\end{figure*}

\begin{figure*}
\begin{minipage}{0.32\textwidth}
\includegraphics[trim=0in 0in .4in .4in, clip, width=\textwidth]{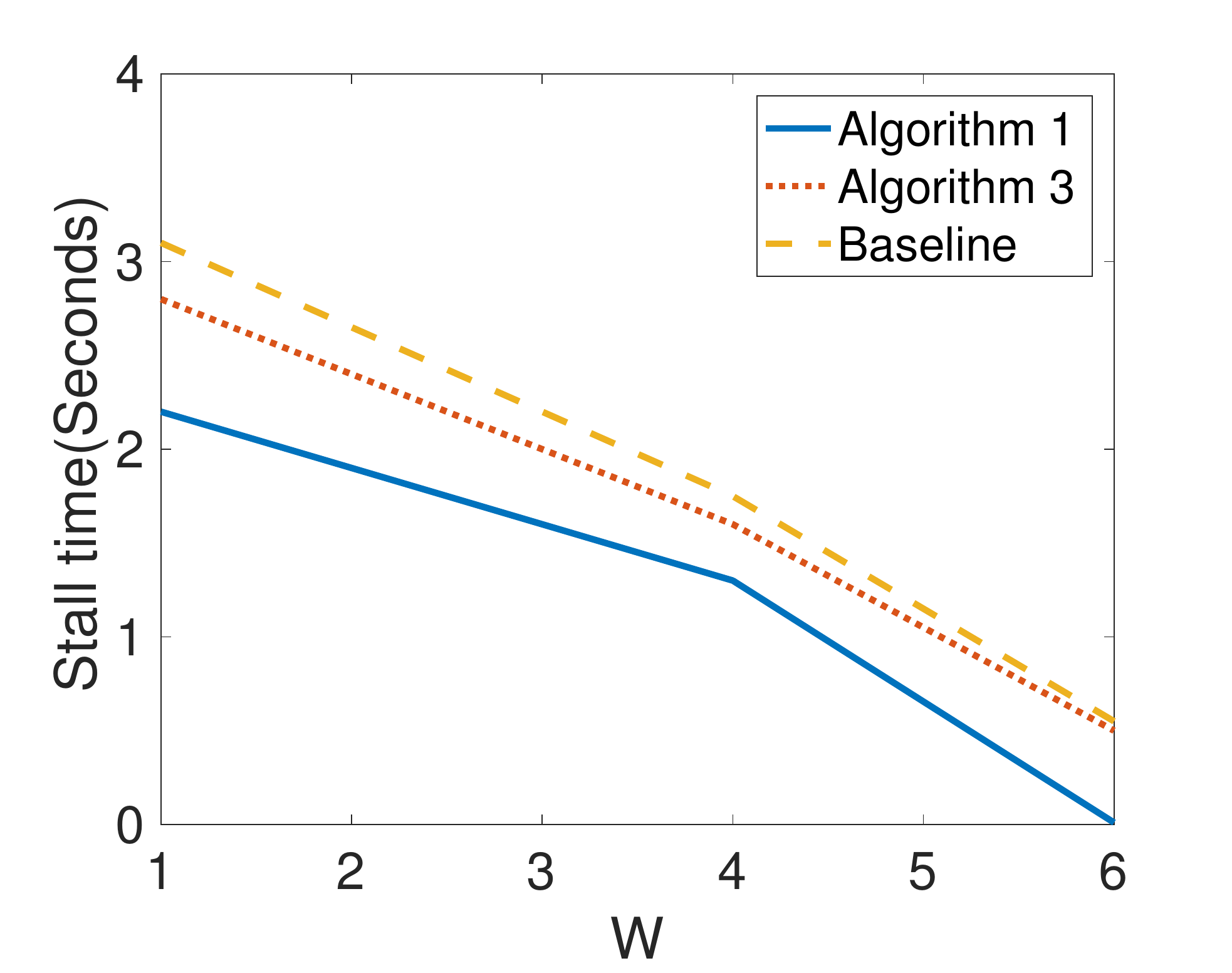}
\vspace{-0.3in}
\caption{The variation of the stall time with the window size W for $\beta=0.8, p=0.25, \gamma=0.1$.}
\label{fig:stall}
\vspace{-0.2in}
\end{minipage}\hfill
\begin{minipage}{0.32\textwidth}
\includegraphics[width=\textwidth]{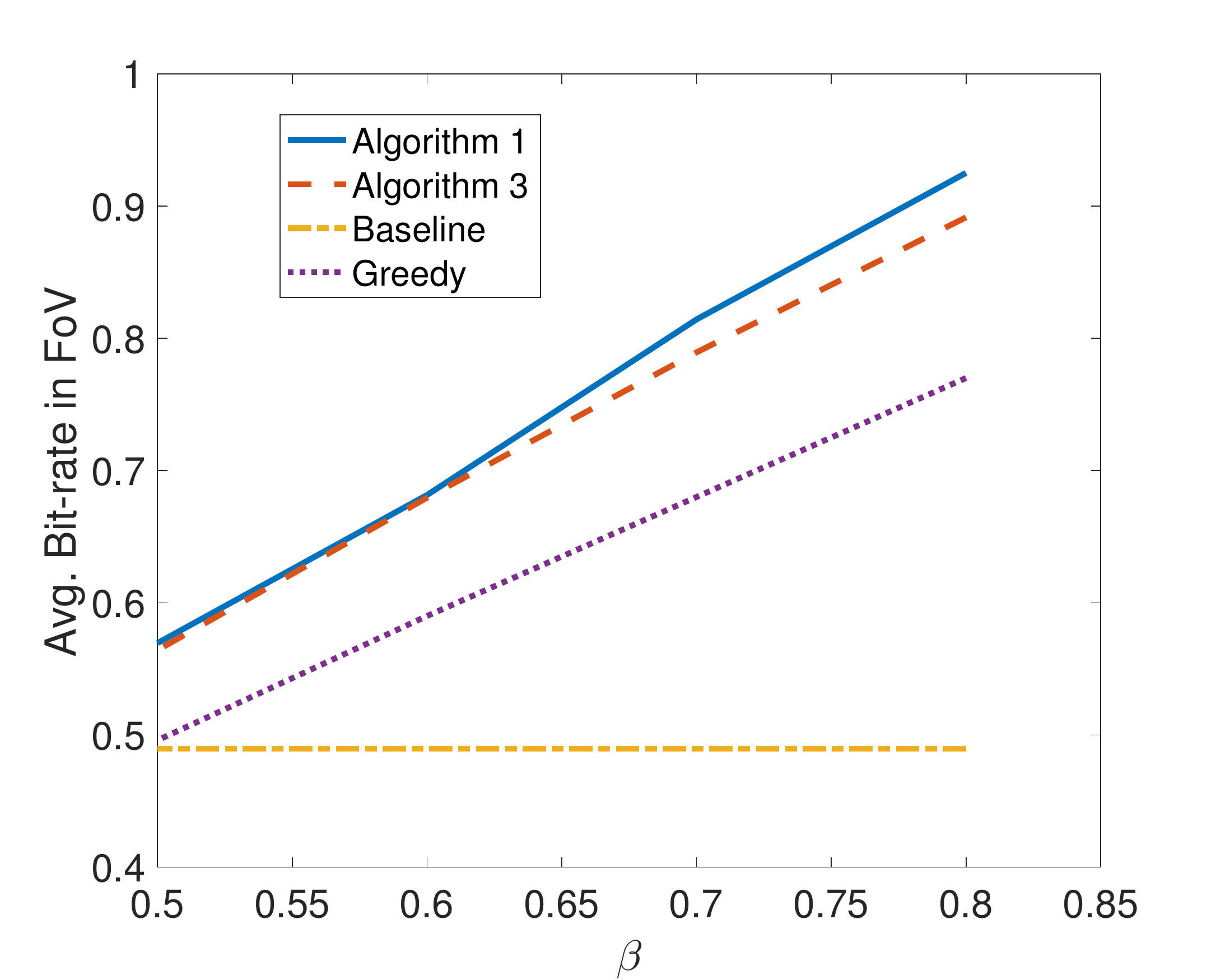}
\vspace{-0.2in}
\caption{The variation of the average bit-rate of the of the downloaded tiles with $\beta$ within the FoV for $p=0.25, W=4$.}
\vspace{-0.2in}
\label{fig:avg_rate}
\end{minipage}\hfill
\begin{minipage}{0.32\textwidth}
\includegraphics[width=\textwidth]{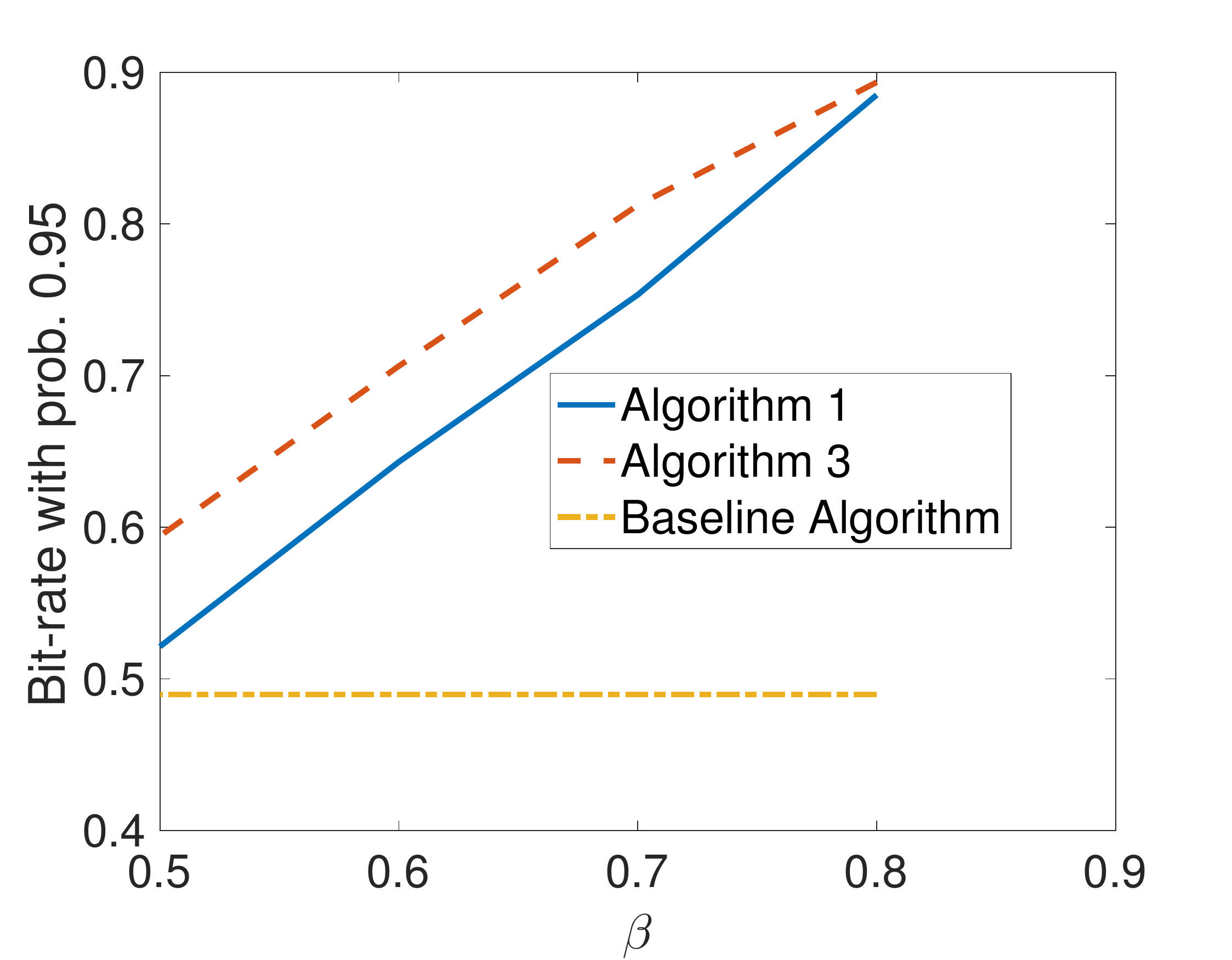}
\vspace{-0.2in}
\caption{The average of the minimum rate among the tiles within the set $\mathcal{A}_{\alpha,k}$ over the chunks for $p=0.25, W=4, \alpha=0.95$ and different values of $\beta$. Greedy algorithm gives $0$ rate.}
\vspace{-0.2in}
\label{fig:robust}
\end{minipage}
\end{figure*}

%\begin{figure*}
%\begin{minipage}{0.48\textwidth}
%\includegraphics[width=0.3in]{gammavsestimationerror-eps-converted-to.pdf}
%\caption{Variation of the optimal value with bandwidth prediction error}
%\end{minipage}\hfill
%\begin{minipage}{0.48\textwidth}
%\includegraphics[width=0.3in]{qualityvariation-eps-converted-to.pdf}
%\end{minipage}
%\end{figure*}
\subsubsection{Impact of $\gamma$} Note from $\mathcal{P}_e$ (cf. (\ref{eq:pe})) that $\gamma$ is the weight corresponding to the expected sum of the rates of the tiles in the FoV. Fig.~\ref{fig:gamma} shows the variation of the objective value  (cf. (\ref{eq:pe})) obtained by different algorithms with $\gamma$. As we mentioned in Section~\ref{sec:qos} a lower value of $\gamma$ indicates that the QoE will be mostly governed by the minimum rate among the tiles in the viewing area (cf. (\ref{eq:pe})).  Higher value of $\gamma$ indicates that the QoE is governed by the average rate among the tiles in the viewing area rather than the minimum rate. 

As $\gamma$ increases, the objective value increases since the weight increases. Algorithm 1 outperforms the other algorithm by at least 20\%.  When $\gamma$ is high it outperforms by more than 30\%. Algorithm 1 is obtained from the solution of the relaxed version of $\mathcal{P}_e$, hence, it performs well  in maximizing the objective. Though our proposed Algorithm~\ref{algo3} is not meant for maximizing the expected QoE it still works better compared to  the greedy algorithm  except for lower values of $\gamma$. This is because when $\gamma$ is low, the optimal algorithm should  fetch the the tiles at the same quality. Thus, the greedy algorithm performs better. However, as $\gamma$ increases, an optimal algorithm should also fetch the tiles at different qualities. Hence, our proposed Algorithm 3 works better compared to the greedy algorithm. Since the greedy algorithm always fetches the tiles which have the highest likelihood to be the part of the FoV in the same quality, its performance is independent of $\gamma$. Algorithm 3 always outperforms the baseline algorithm.

However, Fig.~\ref{fig:gammavserror} shows that as the FoV prediction error increases  (or, $\beta$ decreases) the greedy algorithm performs poorly since the probability to fetch the wrong tiles increases. Now, Algorithm 3 performs better compared to the greedy algorithm even for small values of $\gamma$. {\em The above shows the strength of our proposed algorithms compared to the greedy algorithm and the baseline approach.} Fig.~\ref{fig:gammavserror} also shows that the increase in the variance of the FoV prediction also decreases the objective value. Algorithm 1 again outperforms the algorithms by more than 15\%  even when $\gamma$ is low.

\subsubsection{Stall time} Fig.~\ref{fig:stall} shows the variation of the stall time with the sliding window size $W$. As $W$ increases, one can get the estimate for a longer amount of time. Thus, the stall time decreases with $W$. Since the greedy algorithm does not fetch the base layers of all the tiers, hence, if the FoV consists of one of those tiles black out will occur. {\em We, thus, do not consider the stall time corresponding to the greedy algorithm.} Algorithms 1 and 3 both outperform the Baseline algorithm as those algorithms do not need to fetch the all the tiles of the same quality unlike the Baseline algorithm.

\begin{figure*}
\begin{minipage}{0.32\textwidth}
\includegraphics[trim=0.1in 0in .5in .2in, clip,width=\textwidth]{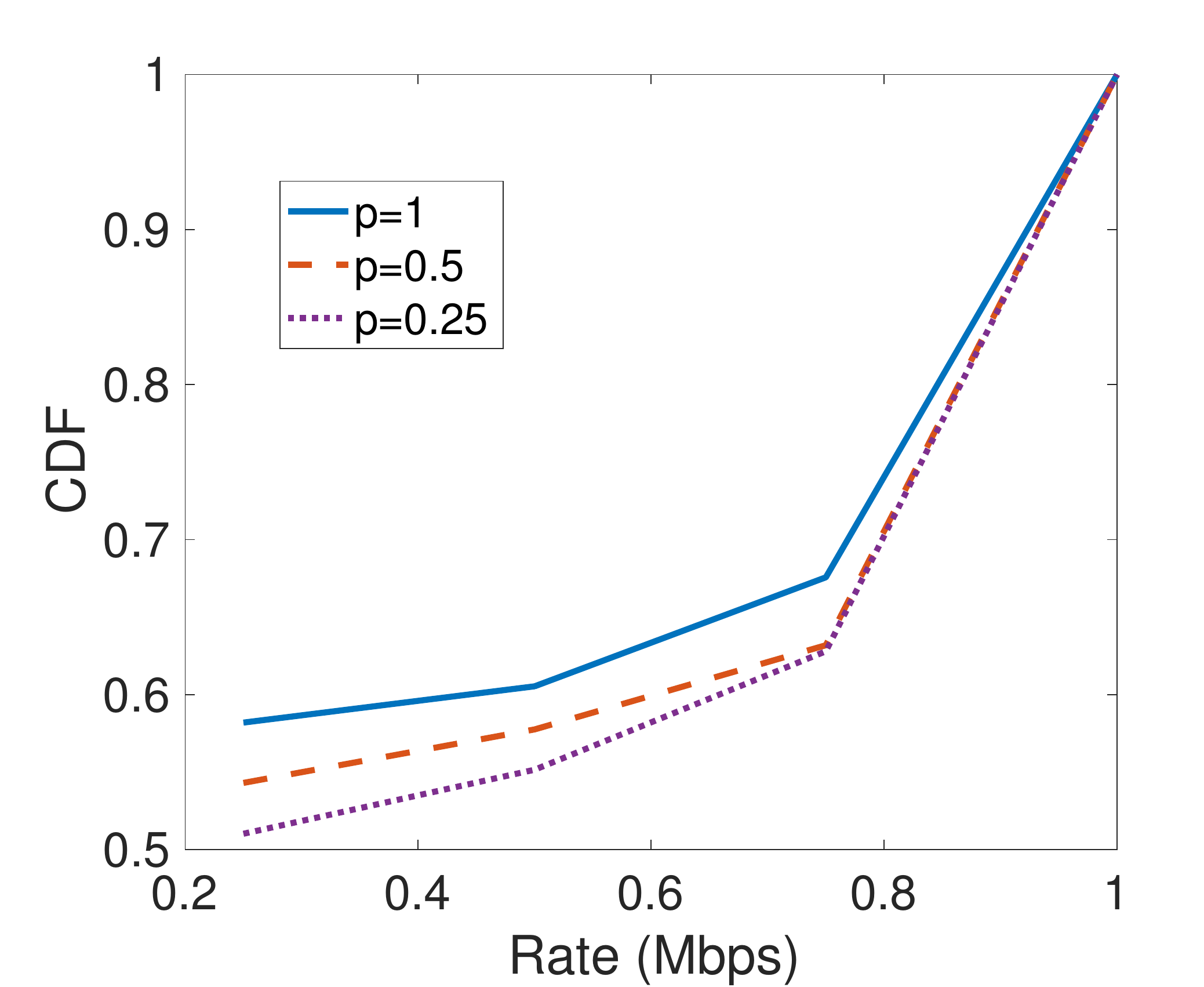}
\vspace{-0.2in}
\caption{The distribution of the bitrates of the downloaded tiles  for different values of $p$ ($\beta=0.8, W=2$).}
\vspace{-0.2in}
\label{fig:bandwidtherror}
\end{minipage}\hfill
\begin{minipage}{0.32\textwidth}
\includegraphics[trim=0.1in 0in .6in .2in, clip, width=\textwidth]{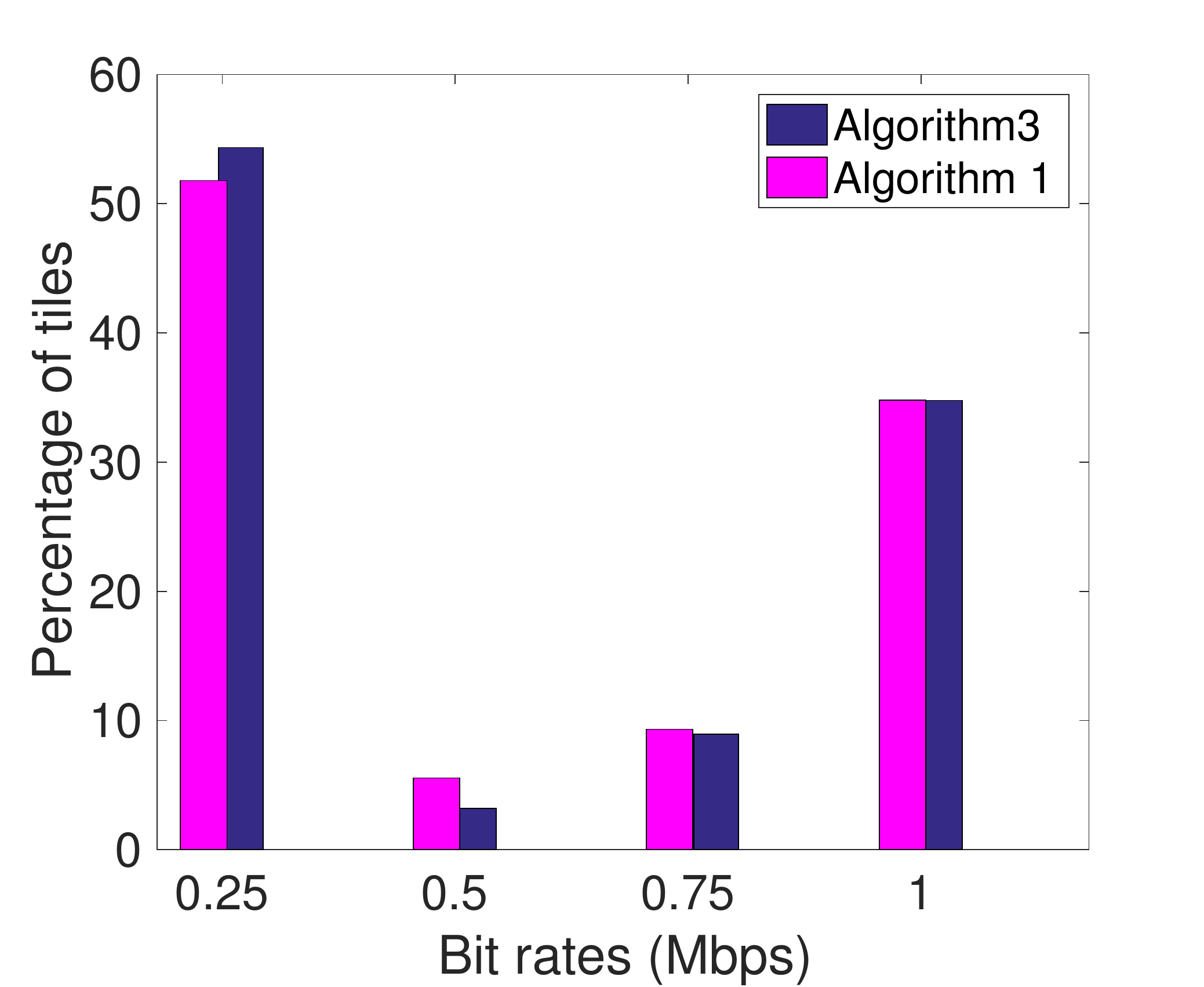}
\vspace{-0.2in}
\caption{The distribution of the bitrates of the downloaded tiles for $\beta=0.8, p=0.25$ and $W=2$.}
\vspace{-0.2in}
\label{fig:rate}
\end{minipage}\hfill
\begin{minipage}{0.33\textwidth}
\includegraphics[width=\textwidth]{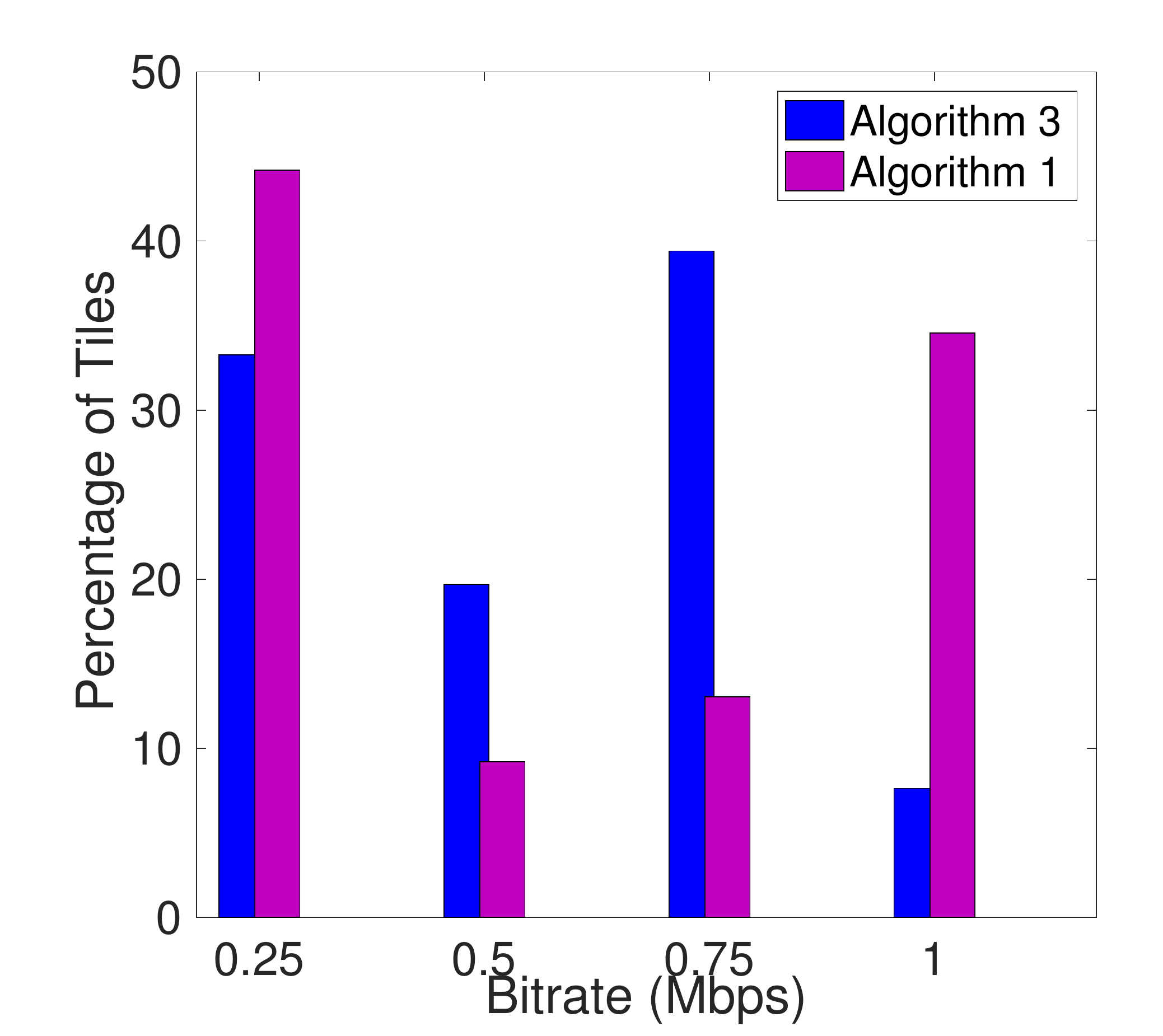}
\vspace{-0.2in}
\caption{The distribution of the bitrates of the downloaded tiles for higher uncertainty ($\beta=0.6$, $p=0.25$, and $W=2$).}
\vspace{-0.2in}
\label{fig:FoVerror}
\end{minipage}
\end{figure*}
\subsubsection{Impact on the average bit-rate within the FoV}
Fig.~\ref{fig:avg_rate} shows the average bit-rate within the FoV. Note that compared to the Baseline algorithm and the greedy one, our proposed algorithms (Algorithms~\ref{algo1} and \ref{algo3})  provide  significant improvement of the bit-rates. This is because Algorithms~\ref{algo1} and \ref{algo3} are obtained by maximizing the QoE metrics. The average bit-rate of all the algorithms decreases as $\beta$ increases, {\it i.e.}, as the FoV uncertainty increases instead of fetching smaller number of tiles at the highest qualities, more tiles needed to be fetched.   The baseline algorithm always fetches all the tiles at the same quality, thus, the average bit-rate attained by this algorithm is independent of the error.

\subsubsection{Impact on the Guaranteed rate} Fig.~\ref{fig:robust} gives the minimum rate among the tiles $\mathcal{A}_{\alpha,k}$ as a function of $\beta$. Recall from (\ref{eq:probust}) that  the FoV will be of at least this rate with probability $\alpha$. Since Algorithm 3 is the optimal (Theorem~\ref{thm:optimal}) for such a QoE metric, thus, it outperforms all the other algorithms. Algorithm 1 performs well   for higher values of $\beta$ (i.e. less uncertainty regarding the FoV). However, when $\beta$ is low (or, uncertainty of the FoV increases), Algorithm 3 gives a significantly higher rate compared to Algorithm 1. Algorithm~\ref{algo1} does not guarantee that the download rates are similar across the tiles.  When $\beta$ is low, more tiles are needed to be fetched at similar qualities. Algorithm~\ref{algo3} makes sure that the {\em quality variation will be poor with very low probability}, hence, it gives a higher guaranteed rate compared to Algorithm~\ref{algo1}. The baseline algorithm gives a constant rate as it is independent of $\beta$. Algorithm~\ref{algo3} outperforms the Baseline algorithm by more than 25\% even when $\beta$ is low as Algorithm~\ref{algo3} can provide higher guaranteed rates by not fetching some tiles outside the set $\mathcal{A}_{\alpha,k}$ at a higher rates. Algorithm 1 also outperforms the Baseline Algorithm as it does not fetch all the tiles with the same quality. However, the difference decreases as $\beta$ decreases. The greedy algorithm gives a $0$ rate as it does not fetch all the tiles within $\mathcal{A}_{\alpha,k}$, thus, it can not provide any guarantee on the rate with probability $\alpha$.

\subsubsection{Impact of Bandwidth Prediction Error}
We, next, study the impact of bandwidth prediction error on the performance of the algorithms. Specifically, we vary $p$; higher $p$ denotes higher variance. Fig.~\ref{fig:bandwidtherror} shows the cumulative distribution of the tiles over the bit-rates. Fig.~\ref{fig:bandwidtherror} shows  that for a higher $p$, more lower quality tiles are fetched. When the estimation error is high, the stall time increases whenever it fetches higher quality tiles which result into decreasing the qualities of the future tiles.

\subsubsection{Distribution of Rates over the tiles} Fig.~\ref{fig:rate} shows the distributions of the downloaded bitrates of tiles of our proposed algorithms. Note that both the Algorithms 1 and 3 fetch the tiles which have very low probability of being the part of the FoV  at a lower quality. That is why the lower quality tiles contribute to a significant portion. However, our analysis shows that the high quality tiles  percentage is also significant. This is because the tiles are fetched at the highest rates which are part of the FoV with a higher probability. %Most of the tiles are mostly downloaded either at the lowest rates or the highest rates. Intuitively, since the utility function is linear, thus, an optimal algorithm tries to fetch the all the higher quality tiles of a chunk if possible before fetching the tiles of future chunks.

\subsubsection{Impact of FoV prediction error}
As the variance in the FoV prediction increases, Fig.~\ref{fig:FoVerror} shows that more  tiles are downloaded in the lower quality by our proposed algorithms Algorithm~\ref{algo1} and \ref{algo3}. However, since Algorithm 3 makes sure that the quality across the tiles is consistent with a higher probability, it has to fetch a higher number of tiles at similar qualities. Thus, more tiles are fetched in a relatively low qualities compared to Algorithm 1. Algorithm 1 fetches the  tiles which have higher probability to be a part of the FoV at the highest qualities, however, the qualities may not be consistent with the high probability. This shows that the Algorithm 3 can provide more consistent viewing experience across the tiles if the prediction is not very accurate.

\section{Future Work}
In this paper, we provide two different QoE metrics. We proposed algorithms to maximize each of the QoE metrics. The characterization of the algorithm which can maximize the convex combinations of the two QoE metrics is a work for the future. We also did not consider the quality variation across the chunks in the viewing area in the QoE. The characterization of the QoE metric for the above parameter is also a work for the future. The implementation of our proposed algorithms in the practical system is also left for the future. 
\bibliographystyle{IEEEtran}
\bibliography{360degreeref}
\appendix
\subsection{Proof of Theorem~\ref{thm:optimal}}\label{sec:proof}
Suppose $R_{i,k}$, $\gamma_k$ for $k=1,\ldots, K$  are obtained using Algorithm~\ref{algo3}. We prove the theorem using Induction. 

For $K=1$:
Note that Algorithm~\ref{algo3} always checks whether the stall time increases before fetching higher quality tiles of the chunk. Hence, Algorithm~\ref{algo3} gives the highest possible rate by keeping the stall time minimum. Since $\lambda$ is large enough, thus, it is optimal for $K=1$. % for all $i\in \mathcal{A}_{\alpha,k}$ and $k\in \{1,\ldots, K\}$. 

Now suppose that it is true for $K=k$, we now show that it is true for $K=k+1$. Suppose not, {\it i.e.}, it is not optimal for $K=k+1$. 

Specifically, there exists $R^{\prime}_{i,j}, \gamma^{\prime}_{j}$ for $j=1,\ldots, k+1$ which can give a higher objective value compared to Algorithm~\ref{algo3}. Note that since $\lambda$ is large enough, thus, the stall time should not increase compared to the solution $R_{i,k}$. Note that $\gamma^{\prime}_j=\min\{R^{\prime}_{i,j}: i\in \mathcal{A}_{\alpha,j}\}$.
%
%The total utility obtained from the solution given by Algorithm~\ref{algo3} is
%\begin{align}
%\sum_{j=1}^{k+1}R_{i,j}
%\end{align}
% The total utility obtained from the solution $R^{\prime}_{i,j}$ is 
% \begin{align}
%\sum_{j=1}^{k+1} R^{\prime}_{i,j}
%\end{align}
%Hence, we must have 
%\begin{align}
%\sum_{j=1}^{k+1}R^{\prime}_{i,j}>\sum_{j=1}^{k+1}R_{i,j}.
%\end{align}
%Since the statement is true for $K=k$, we must have
%\begin{align}
%\sum_{j=1}^{k}R^{\prime}_{i,j}\leq\sum_{j=1}^{k}R_{i,j}.
%\end{align}
Because of the linear utility, the utility attained by Algorithm~\ref{algo3} is 
\begin{align}
\sum_{j=1}^{k+1}\gamma_j.
\end{align}
Similarly, the utility attained by the solution $R^{\prime}_{i,k}$ is
\begin{align}
\sum_{j=1}^{k+1}\gamma^{\prime}_j
\end{align}
Note that
\begin{align}\label{eq:gam1}
\sum_{j=1}^{k+1}\gamma^{\prime}_j>\sum_{j=1}^{k+1}\gamma_j.
\end{align}

Since the statement is true for $K=k$,  we must have 
\begin{align}\label{eq:gam2}
\sum_{j=1}^{k}\gamma^{\prime}_j\leq \sum_{j=1}^{k}\gamma_j. 
\end{align}
%\end{align}However, from 
%\begin{align}\label{eq:gam}
%\gamma^{\prime}_{k+1}+\sum_{j=1}^{k}\gamma^{\prime}_j>\gamma_{k+1}+\sum_{j=1}^{k}\gamma_j
%\end{align}
The total bits to be downloaded for the solution $R^{\prime}_{i,j}$ is
\begin{align}\label{eq:bit}
L[\sum_{j=1}^{k+1}|\mathcal{A}_{\alpha,j}|\gamma^{\prime}_j+\sum_{j=1}^{k+1}R_{0}|\mathcal{A}^{C}_{\alpha,j}|]
\end{align}
Since it does not increase the stall time, thus, the above number of bits can be downloaded within the deadline $t_{k+1}$. 
Now,
\begin{align}\label{eq:inequal}
\gamma^{\prime}_{k+1}|\mathcal{A}_{\alpha,k+1}|-\sum_{j=1}^{k}(\gamma_j-\gamma^{\prime}_j)|\mathcal{A}_{\alpha,k+1}|\nonumber\\
\leq \gamma^{\prime}_{k+1}|\mathcal{A}_{\alpha,k+1}|-\sum_{j=1}^{k}(\gamma_j-\gamma^{\prime}_j)|\mathcal{A}_{\alpha,k}|
\end{align}
where the last inequality follows from (\ref{eq:gam2}) and Assumption~\ref{assum:ordering}. Hence, the total number of bits that can be downloaded if $\gamma_{k+1}=\gamma^{\prime}_{k+1}-\sum_{j=1}^{k}(\gamma_j-\gamma^{\prime}_j)$ is given by 
\begin{align}
& L[\gamma^{\prime}_{k+1}|\mathcal{A}_{\alpha,k+1}|-\sum_{j=1}^{k}(\gamma_j-\gamma^{\prime}_j)|\mathcal{A}_{\alpha,k+1}|+\sum_{j=1}^{k+1}R_{0}|\mathcal{A}^{C}_{\alpha,j}|\nonumber\\& +\sum_{j=1}^{k}\gamma_j|\mathcal{A}_{\alpha,j}|]
\end{align}
Hence, from (\ref{eq:inequal}) and Assumption~\ref{assum:ordering} the above expression can be upper bounded by
\begin{align}
& L[\gamma^{\prime}_{k+1}|\mathcal{A}_{\alpha,k+1}|-\sum_{j=1}^{k}\gamma_j|\mathcal{A}_{\alpha,k}|+\sum_{j=1}^{k}\gamma_j|\mathcal{A}_{\alpha,j}|+\nonumber\\& \sum_{j=1}^{k+1}\gamma^{\prime}_j|\mathcal{A}_{\alpha,j}|+\sum_{j=1}^{k+1}R_{0}|\mathcal{A}^{C}_{\alpha,j}|]\nonumber\\
& \leq L[\gamma^{\prime}_{k+1}|\mathcal{A}_{\alpha,k+1}|+\sum_{j=1}^{k}\gamma^{\prime}_j|\mathcal{A}_{\alpha,j}|+\sum_{j=1}^{k+1}R_{0}|\mathcal{A}^{C}_{\alpha,j}|]
\end{align}
where the last inequality again follows from Assumption~\ref{assum:ordering}. However, the above expression is exactly equal to (\ref{eq:bit}), thus, the utility attained in the $k+1$th chunk without increasing the stall time is $\gamma^{\prime}_{k+1}-\sum_{j=1}^{k}(\gamma_j-\gamma^{\prime}_j)$. Thus, the utility attained by the solution given in Algorithm~\ref{algo3} is at least
\begin{align}
\sum_{j=1}^{k+1}(\gamma_{j}-\gamma_j+\gamma^{\prime}_j)+\gamma^{\prime}_{k+1}=\sum_{j=1}^{k+1}\gamma^{\prime}_{j}
\end{align}
which contradicts (\ref{eq:gam1}). Hence, the statement is true for $K=k+1$ assuming that it is true for $K=k$. 

Thus, from the principle of mathematical induction the result follows. \qed

%However, if the quality of service over the tiles varies much for a, the user experience may degrade if an user
\end{document}